\documentclass[a4paper, 12pt]{article}
\usepackage{amsmath}
\usepackage{amsfonts}
\usepackage{amssymb}
\usepackage{amsthm}
\usepackage{bm}
\usepackage{graphics,color}
\newtheorem{theorem}{Theorem}
\newtheorem{lemma}{Lemma}
\newtheorem{cor}{Corollary}
\newtheorem{prop}{Proposition}
\newtheorem{conj}{Conjecture}
\theoremstyle{remark}
\newtheorem{remark}{Remark}
\theoremstyle{definition}
\newtheorem{defn}{Definition}

\newcommand{\cir}{|\mbox{Circ}\rangle_n}

\newcommand{\ellip}{|\mbox{Elliptic}(\theta)\rangle}
\newcommand{\ellipt}{\langle\mbox{Elliptic}(\theta)|}
\newcommand{\ex}[1]{\mathbb{E}\left\{ #1 \right\}}
\newcommand{\pr}[1]{\mathbb{P}\left\{ #1 \right\}}
\newcommand{\hf}[1][1]{\frac{#1}{2}}

\newcommand{\di}{\,\mathrm{d}}
\newcommand{\ii}{\mathrm{i}}
\renewcommand{\vec}[1]{\mbox{\boldmath$#1$}}
\newcommand{\svec}[1]{\mbox{\boldmath$\scriptstyle #1$}}
\begin{document}
\title{The Divine Clockwork: Bohr's correspondence principle and Nelson's stochastic mechanics for the atomic elliptic state.}
\author{Richard Durran $\quad$  Andrew Neate$\quad$  Aubrey Truman\\
\small Department of Mathematics, Swansea University, \\
\small Singleton Park, Swansea, SA2 8PP, Wales, UK.}
\maketitle
\begin{abstract}
We consider the Bohr correspondence limit of the Schr\"{o}dinger wave function for an atomic elliptic state. We analyse this limit in the context of Nelson's stochastic mechanics, exposing an underlying deterministic dynamical system in which trajectories converge to Keplerian motion on an ellipse. This solves the long standing problem of obtaining Kepler's laws of planetary motion in a quantum mechanical setting. In this quantum mechanical setting, local mild instabilities occur in the Keplerian orbit for eccentricities greater than $1/ \sqrt{2}$ which do not occur classically.
% which could be detected in a suitable experiment if the atomic elliptic state could be realised in the laboratory.
%If the atomic elliptic state could be produced in the laboratory, this would give a possible experimental test of Nelson's stochastic mechanics. These instabilities  may also explain in some detail when and how the orbiting quantum particle changes its state by emitting or absorbing photons. 
\end{abstract}

%%%%%%%%%%%%%%%%%%%%%%%%%%%%%%%%%%%%%%%%%%%%%%%%%%%%%%%%%%%%%%%%%%%%%%%
%%%%%%%%%%%%%%%%%%%%%%%%%%%%%%%%%%%%%%%%%%%%%%%%%%%%%%%%%%%%%%%%%%%%%%%
\section{Background}
\subsection{Introduction}
One of the longest standing problems in quantum mechanics is to obtain Kepler's laws of planetary motion by taking the Bohr correspondence limit of a suitable Schr\"{o}dinger wave function for an  atomic elliptic state. The Bohr correspondence limit of a quantum state is found by allowing the quantum numbers to become infinite while  $\hbar$  becomes vanishingly small with the energies and angular momenta of the system fixed with physical values of order one \cite{Bohr}. We have approached this quantum mechanical problem from the perspective of Nelson's stochastic mechanics.

In 1966, Edward Nelson produced a new formulation of non relativistic quantum mechanics in terms of diffusion processes satisfying the Nelson-Newton law, a stochastic version of Newton's second law of motion \cite{Nelson, MR0214150, MR783254}.  This work was based on the  assumption that every particle of mass $m$ moves in a random environment driven by a Brownian motion with diffusion constant $\hbar/2m$.  Some mathematical physicists felt that Nelson's formulation merely recapitulated the standard results of Schr\"{o}dinger quantum mechanics without explaining the physical agency responsible for the Brownian noise in the system, and so the theory did not prove popular. However, the theory does give a new mental picture of the quantum world in which one can calculate the statistics of new observables, like first hitting times, which should be open to measurement \cite{MR1182593}. Moreover, within this formulation, Nelson showed that the Nelson-Newton law leads uniquely to Schr\"{o}dinger quantum mechanics, advocating the view that the Schr\"{o}dinger equation is a way of linearising a very complicated non-linear system of equations equivalent to the Nelson-Newton law.

In this paper we will use Nelson's stochastic mechanics to take the correspondence limit of the atomic elliptic state for the Coulomb potential and show that this limit gives Kepler's laws of planetary motion.

 %%%%%%%%%%%%%%%%%%%%%%%%%%%%%%%%%%%%%%%%%%%%%%%%%%%%%%%%%%%%%%%%%%%%%%%
 %%%%%%%%%%%%%%%%%%%%%%%%%%%%%%%%%%%%%%%%%%%%%%%%%%%%%%%%%%%%%%%%%%%%%%%
\subsection{Nelson's stochastic mechanics}
We begin with a brief account of Nelson's stochastic mechanics based on the summaries in  \cite{MR960159, MR1020069}. Let $\vec{X}$ be the position of a particle of unit mass diffusing in $\mathbb{R}^d$ according to an It\^{o} stochastic differential equation with diffusion constant $\epsilon^2/2$,
\begin{equation}\label{diffusion}
\di \vec{X}(t) = \vec{b}(\vec{X}(t),t)\di t  +\epsilon \di \vec{B}(t),\quad t>0, \quad\vec{X}(0) = \vec{x},
\end{equation}
where $\vec{B}=(B_1,B_2,\ldots,B_d)$ is a $d$-dimensional Brownian motion process on a probability space $(\Omega,\mathcal{A}_t,\mathbb{P})$, with covariance,
\[\ex{B_i(t)B_j(s)}=\delta_{ij}(t\wedge s),\]
where $\mathbb{E}$ denotes expectation with respect to the measure $\mathbb{P}$.

We assume that the process $\vec{X}$ has a smooth density function $\rho(\vec{y},t)$ such that for any $A\in\mathcal{A}_t$,
\[\pr{\vec{X}(t)\in A} = \int_{\svec{y}\in A}\rho(\vec{y},t)\di\vec{y}.\]
Then the density $\rho$ must satisfy the forward Kolmogorov equation,
\begin{equation}\label{forward}
\frac{\partial \rho}{\partial t}(\vec{y},t) = \nabla\cdot\left(\hf[\epsilon^2]\nabla\rho(\vec{y},t) -\rho(\vec{y},t)\vec{ b}(\vec{y},t)\right).
\end{equation}

We define the mean forward and backward derivatives $D_{\pm}$ by,
\[D_{\pm} f(\vec{X}(t),t) := \lim\limits_{h\downarrow 0} \ex{\left.\frac{f(\vec{X}(t\pm h),t\pm h)-f(\vec{X}(t),t)}{\pm h}\right|\vec{X}(t)}.\]
Assuming some mild regularity conditions on $f$, It\^{o}'s formula gives,
\[D_+f(\vec{X}(t),t) = \left(\frac{\partial}{\partial t}+\vec{b}(\vec{X}(t),t)\cdot\nabla +\hf[\epsilon^2]\Delta\right) f(\vec{X}(t),t).\]
In  particular, it follows that the mean forward velocity is simply,
\[ \vec{b}_+(\vec{X}(t),t):=D_+ \vec{X}(t)  = \vec{b}(\vec{X}(t),t).\]
Furthermore, it can be shown that for $g = g(\vec{X}(t),t)$ and $h = h(\vec{X}(t),t)$,
\[\frac{\di}{\di t} \ex{gh} = \ex{g D_-h}+\ex{hD_+g},\]
and it follows, using integration by parts, that for $f =f(\vec{X}(t),t)$,
\[D_-f = \left(\frac{\partial }{\partial t} + \left( \vec{b}_+-\epsilon^2\nabla\ln\rho\right)\cdot\nabla -\hf[\epsilon^2]\Delta \right)f.\]
This gives the mean backward velocity as,
\[\vec{b}_-(\vec{X}(t),t):=D_-\vec{X}(t) = \vec{b}_+(\vec{X}(t),t) - \epsilon^2\nabla\ln\rho(\vec{X}(t),t).\]
We can also introduce the osmotic velocity $\vec{u}$ and the current velocity $\vec{v}$ such that,
\begin{equation}\label{osmotic}
\vec{u}: = \hf(\vec{b}_+-\vec{b}_-),\qquad \vec{v} := \hf(\vec{b}_++\vec{b}_-).
\end{equation}
If we  now define the stochastic acceleration by,
\[\vec{a}(\vec{X}(t),t) = \hf\left(D_+D_-+D_-D_+\right)\vec{X}(t),\]
then it follows that,
\begin{equation}\label{acc}
\vec{a}(\vec{X}(t),t) = \left(\frac{\partial \vec{v}}{\partial t} +(\vec{v}\cdot\nabla)\vec{v}-(\vec{u}\cdot\nabla)\vec{u} -\hf[\epsilon^2]\Delta \vec{u}\right)(\vec{X}(t),t).
\end{equation}
We now show how this apparently impenetrable expression takes on a deep significance in connection with the Schr\"{o}dinger equation.

Consider a quantum mechanical particle of unit mass  in $\mathbb{R}^d$ subject to the force field $-\nabla V$  where $V$ is some real valued potential. The corresponding complex valued wave function $\psi$ satisfies the Schr\"{o}dinger equation,
\begin{equation}\label{schrodinger}
\ii \hbar \frac{\partial\psi}{\partial t} = -\hf[\hbar^2]\Delta\psi +V\psi,
\end{equation}
where (letting $\psi^*$ denote the complex conjugate of $\psi$) the quantum mechanical particle density is,
\[\rho = |\psi|^2 = \psi \psi^*.\]
 Multiplying through equation (\ref{schrodinger})  by $\psi^*$ gives the equivalent Schr\"{o}dinger equation,
\begin{equation} \label{conjugate}
\ii\hbar \psi^*\frac{\partial\psi}{\partial t} = -\hf[\hbar^2]\psi^*\Delta\psi + V|\psi|^2.
\end{equation}
Equating imaginary parts in (\ref{conjugate}) gives the continuity equation,
\begin{equation}\label{continuity}
\frac{\partial\rho}{\partial t} + \nabla\cdot j = 0,
\end{equation}
where $j$ is the probability current,
 \[j= \frac{\hbar}{2 \ii}(\psi^*\nabla\psi-\psi\nabla\psi^*).\]
If we now write,
\[\psi = \exp(R+\ii S),\]
where $R$ and $S$ are real valued functions of space and time, then it follows that,
\[\rho = \exp(2R),\qquad j = \hbar\rho \nabla S.\]
The continuity equation (\ref{continuity}) now becomes,
\[
\frac{\partial \rho}{\partial t} = \nabla\cdot\left(-\hbar\nabla Se^{2R}\right)  = \nabla\cdot\left(\hf[\hbar]\nabla e^{2R}-\hbar e^{2R}\nabla(R+S)\right).\]
Therefore, if we let,
\[\epsilon ^2 = \hbar, \qquad \vec{b} = \hbar\nabla(R+S),\]
then we have,
\[\frac{\partial \rho}{\partial t} = \nabla\cdot\left(\hf[\epsilon^2]\nabla\rho-\rho\,\vec{b}\right).\]
This is the forward Kolmogorov equation (\ref{forward}) for a diffusion $\vec{X}(t)$
with forward and backward velocities,
\begin{equation}
\label{fwdbckwd}
\vec{b}_+(\vec{X}(t),t) = \epsilon^2\nabla(R+S),\qquad \vec{b}_-(\vec{X}(t),t) = \epsilon^2\nabla(S-R).
\end{equation}
 A tedious calculation from equations (\ref{osmotic}), (\ref{acc}) and (\ref{fwdbckwd}) gives the stochastic acceleration in terms of $R$ and $S$ as,
\begin{equation}\label{acceleration}
\vec{a}(\vec{X}(t),t)  = -\nabla \left(-\epsilon^2\frac{\partial S}{\partial t} + \hf[\epsilon^4]\left(|\nabla R|^2-|\nabla S|^2 + \Delta R\right)\right)(\vec{X}(t),t).
\end{equation}
However, equating real parts of equation (\ref{conjugate}) gives,
\begin{equation}\label{acc2}
\frac{\partial S}{\partial t} = \hf[\epsilon^2]\left(|\nabla R|^2 -|\nabla S|^2+\Delta R\right)-\frac{V}{\epsilon^2}.
\end{equation}
Therefore, combining equations (\ref{acceleration}) and (\ref{acc2}) gives,
\begin{equation}\label{newton}
\vec{a}(\vec{X}(t),t) = -\nabla V(\vec{X}(t),t),
\end{equation}
where $V$ is the potential in the Schr\"{o}dinger equation. Equation (\ref{newton}) is  the Nelson-Newton law for a particle of unit mass.

The Nelson-Newton law argues that to investigate the Bohr correspondence limit of a quantum state,  we should investigate the Nelson diffusion process in the appropriate limit. 

%%%%%%%%%%%%%%%%%%%%%%%%%%%%%%%%%%%%%%%%%%%%%%%%%%%%%%%%%%%%%%%%%%%%%%%
\subsection{Outline of the paper}
Section 2 starts from Pauli's elegant solution of the hydrogen atom. We follow the results of Lena, Delande and Gay \cite{MR1155697} to obtain the Schr\"{o}dinger wave function corresponding to an atomic elliptic state in $\mathbb{R}^3$ and then find its related Nelson diffusion process.
In Section 3, we derive the formal limiting wave function for this atomic elliptic state according to Bohr's prescription. We derive the corresponding limiting Nelson diffusion and
present simulations of the limiting process showing how the process converges in large times to a particular ellipse which we call the Kepler ellipse. We investigate the behaviour of the  drift for this limiting diffusion and demonstrate that it has a finite jump discontinuity across part of the semi-major axis of the Kepler ellipse. This singularity is the Bohr correspondence limit of the nodal surfaces of the atomic elliptic state wave function.

For the remainder of the paper we restrict our wave functions and diffusion processes to the putative plane of motion $(z=0)$.
In Section 4 we introduce a new set of elliptic coordinates, the non-orthogonal Keplerian elliptic coordinates, and show how these greatly simplify the limiting Nelson diffusion process.

In Section 5, we analyse the two dimensional limiting diffusion process. We discuss the behaviour of the invariant measure for this system and demonstrate that it is sharply peaked on the Kepler ellipse. Taking an idea from earlier work on excursions in stochastic mechanics \cite{MR1189398}, we introduce a similarity transform for the generator of the Nelson diffusion process and use this to discuss the convergence over time of the density for the Nelson diffusion processes to their invariant measures.  We also show that there is a positive probability of hitting the drift singularity in a finite time. Thus, to avoid inventing an artificial boundary condition,  we restrict our path space so that our diffusing particle avoids this singularity and we estimate the probability for this to happen. On the restricted path space we see that the Bohr correspondence limit reduces to the underlying deterministic dynamical system.

In Section 6 we look in detail at this underlying deterministic dynamical system.
We show that this system has the Kepler ellipse as a stable periodic orbit and we derive Kepler's laws of planetary motion for a particle on this orbit.
We identify which paths avoid the singularity and for these paths discuss the asymptotic stability of the Kepler ellipse. 
We show that the invariant density can be used to construct a Lyapunov function for the system.  Using this, we show that  all motions starting outside the Kepler ellipse converge to Keplerian motion on the Kepler ellipse.   For motions starting inside the Kepler ellipse the result is more difficult to prove and we leave this to a future publication. We conclude by highlighting some surprising symmetries within the diffusion process and use these symmetries to show that mild local instabilities occur in the dynamical system for $e>1/\sqrt{2}$ where $e$ is the eccentricity of the Kepler ellipse.

Thus, we see that in $\mathbb{R}^2$, for orbits starting outside the Kepler ellipse and avoiding the singularity, the long time behaviour  of the Bohr correspondence limit of Nelson's stochastic mechanics for atomic elliptic states is Keplerian motion on the Kepler ellipse. We believe that we can extend this result to $\mathbb{R}^3$ but so far we have not found a convenient coordinate system in which to carry out the analysis.

%%%%%%%%%%%%%%%%%%%%%%%%%%%%%%%%%%%%%%%%%%%%%%%%%%%%%%%%%%%%%%%%%%%%%%%
%%%%%%%%%%%%%%%%%%%%%%%%%%%%%%%%%%%%%%%%%%%%%%%%%%%%%%%%%%%%%%%%%%%%%%%
\section{The atomic elliptic state}
\subsection{The wave function for the  atomic elliptic state}
In this section we follow the work of Lena, Delande and Gay \cite{MR1155697}. Recall that the atomic circular state $\cir$ corresponding to a Keplerian circular orbit is given by,
\[\langle \vec{x}\cir = \Psi_{n,n-1,n-1}(\vec{x}),\]
where $\vec{x} = (x,y,z)$ and  $\Psi_{n,l,m}$ (with  $l = 0,1,\ldots,n-1,$ $|m|\le l$ and $m,n \in\mathbb{N} $) is the nodal Schr\"{o}dinger wave function  for the Hamiltonian,
\begin{equation}\label{hamiltonian}
H(\vec{p},\vec{q}) = \hf[\vec{p}^2]-\frac{\mu}{|\vec{q}|}.
\end{equation}
The vectors $\vec{p}=(p_1,p_2,p_3)$ and  $\vec{q} = (q_1,q_2,q_3)$ are the momentum and position operators in cartesian coordinates for the orbiting quantum particle with
$[q_k,p_l] = \ii \hbar \delta_{kl}$ for $k,l = 1,2,3.$
We now work in suitable units so that $\hbar = 1$ and $\mu=1$. Then, for the orbital angular momentum $\vec{L}$, where
\[\vec{L} =(L_1,L_2,L_3)= (\vec{q}\bm{\times} \vec{p}),\qquad \vec{L}^2 = (L_1^2+L_2^2+L_3^2),\]
we have,
\[\vec{L}^2 \Psi_{n,l,m} = l(l+1)\Psi_{n,l,m},\quad L_3 \Psi_{n,l,m} = m\Psi_{n,l,m},\quad H \Psi_{n,l,m} = -\frac{1}{2 n^2}\Psi_{n,l,m}.
 \]

 The state we consider is,
\begin{equation}\label{elliptic} \ellip_n = \exp(-\ii\theta A_2) \cir,
\end{equation}
where,
\[\vec{A} = (A_1,A_2,A_3) = \frac{1}{\sqrt{-2E}} \left(\hf[(\vec{p}\bm{\times} \vec{L}-\vec{L}\bm{\times} \vec{p})]-\frac{\vec{q}}{|\vec{q}|}\right),\]
is the Hamilton-Lenz-Runge vector on the eigenspace of $H$ with eigenvalue $E$. We will show that on the space where
$E = E_n = -1/(2n^2)$, the state $\ellip_n$ as $n\uparrow \infty$ corresponds to the elliptic atomic state with eccentricity $e = \sin\theta$ for some $\theta\in(0,\pi)$.

Pauli proved the fundamental identities,
\[[A_1,A_2] = \ii L_3,\quad [A_2,L_3] = \ii A_1,\quad [L_3,A_1] = \ii A_2,\]
showing that $A_1,A_2,L_3$ (the constants of motion in the 2-dimensional problem) generate the dynamical symmetry group SO(3) for this situation. Now define,
\begin{eqnarray*}
\langle L_3\rangle(\theta) & = & \ellipt L_3\ellip,\\
 \langle A_1\rangle(\theta)&  = & \ellipt A_1\ellip.
 \end{eqnarray*}
It follows that,
\[\frac{\di^2}{\di \theta^2}\langle L_3\rangle (\theta) = -\langle L_3\rangle(\theta),\quad \frac{\di^2}{\di \theta^2}\langle A_1\rangle (\theta) = -\langle A_1\rangle(\theta),\]
giving,
\[\langle L_3\rangle (\theta) = (n-1) \cos(\theta),\quad \langle A_1\rangle (\theta) = (n-1) \sin(\theta).\]
Moreover, if we consider the classical identity relating the eccentricity $e$ of the elliptic orbit to the angular momentum $\vec{L}$ and the energy $E$,
\[e^2 = 1+2 \vec{L}^2 E,\]
and take  the limit as $n\rightarrow\infty$, we find $e^2 = \sin^2(\theta)$ as asserted. A clever argument  using coherent state representations and the Kustaanheimo-Stiefel transformation yields, for $\theta_0 = \arcsin (e)$ where $0<e<1$, to within a multiplicative constant (after reinstating $\hbar$ and the force constant $\mu$),
\begin{equation}\label{elliptic2}
\psi_{\epsilon,n}(\vec{x}):=\langle x |\mbox{Elliptic}(\theta_0)\rangle_n = \exp\left(-\frac{n\mu }{\lambda^2}|\vec{x}|\right) \mathcal{L}_{n-1}(n\nu),
\end{equation}
where
\begin{equation}\label{nu}\nu = \frac{\mu}{\lambda^2}\left(|\vec{x}| - \frac{x}{e} - \frac{\ii y\sqrt{1-e^2}}{e}\right),
\end{equation}
with $\vec{x} = (x,y,z)$, $E_n = -\mu^2/(2\lambda^2)$, $\epsilon^2=\hbar$,  $\lambda = n \epsilon^2$
and $\mathcal{L}_{n-1}$ a Laguerre polynomial.

The nodes of the wave function $\psi_{\epsilon,n}$ are located on a series of $(n-1)$ hyperbolas in the plane $y=0$ each with axis $x$, eccentricity $1/e$ and focus at the origin (see Figure \ref{nodal}).
\begin{figure}
\centering
\resizebox{120mm}{!}{\includegraphics{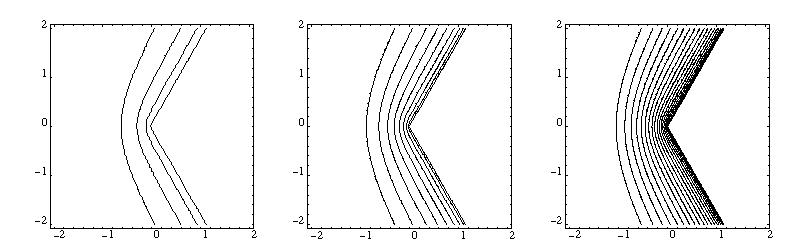}}
\caption{The nodal surfaces in the $(x,z)$ plane for $\psi_{\epsilon,n}$ for $e=0.5$ and $n = 5$, $10$, $20$.}
\label{nodal}
\end{figure}

Associated with the atomic elliptic state  $|\mbox{Elliptic}(\theta_0)\rangle_n$ where $\arcsin \theta_0 = e$, there is a special ellipse which we will call the Kepler ellipse.
\begin{defn}\label{Keplerellipse} The Kepler ellipse is the ellipse in the plane $z=0$ with eccentricity $e$ and semi-major axis $a$ with one focus at the origin given in cylindrical polar coordinates by,
 \begin{equation}\label{kepler}
 \tilde{r} = \frac{a (1-e^2)}{1+e \cos\theta},\quad z=0,
 \end{equation}
 where $a = \lambda^2/\mu$ and  $x =\tilde{r}\cos\theta,$ $y= \tilde{r} \sin\theta$ with $\tilde{r} = \sqrt{x^2+y^2}$.
\end{defn}
We will show that the correspondence limit of the Nelson diffusion for the  atomic elliptic state with wave function $\psi_{\epsilon,n}$ is Keplerian motion on the Kepler ellipse.

\begin{remark}
Throughout this work the letter $e$ will refer to the eccentricity associated with the atomic elliptic state. When we use the exponential function we shall always write $\exp(x)$ rather than $e^x$ to avoid confusion.
\end{remark}

 %%%%%%%%%%%%%%%%%%%%%%%%%%%%%%%%%%%%%%%%%%%%%%%%%%%%%%%%%%%%%%%%%%%%%%%
\subsection{The Nelson diffusion for the atomic elliptic state}
We now derive the Nelson diffusion process $\vec{X}^{\epsilon,n}$ associated with the wave function $\psi_{\epsilon,n}$ for the atomic elliptic state  $|\mbox{Elliptic}(\theta_0)\rangle_n$.
If we write   $\psi_{\epsilon,n} = \exp(R_{\epsilon,n}+\ii S_{\epsilon,n})$ then $\vec{X}^{\epsilon,n}$ satisfies,
\[
\di \vec{X}^{\epsilon,n}(t) = \vec{b}_{\epsilon,n}(\vec{X}^{\epsilon,n}(t))\di t  +\epsilon \di \vec{B}(t),
\]
with,
\[\vec{b}_{\epsilon,n} = \epsilon^2\nabla(R_{\epsilon,n}+S_{\epsilon,n}) .\]
We will now find the drift $\vec{b}_{\epsilon,n}$ for the atomic elliptic state.

The wave function $\psi_{\epsilon,n}(\vec{x})$ satisfies the Schr\"{o}dinger equation,
\[-\hf\epsilon^4\Delta\psi_{\epsilon,n}(\vec{x})-\frac{\mu}{|\vec{x}|}\psi_{\epsilon,n}(\vec{x}) = E_n\psi_{\epsilon,n}(\vec{x}),\]
where $\epsilon^2 = \hbar$, $E_n = -\mu^2/(2\lambda^2)$, $\lambda  = n\epsilon^2$ and $\vec{x} = (x,y,z)$. Defining $\vec{Z}_{\epsilon,n}(\vec{x})$ by,
\begin{equation}\label{Zdefn}
\vec{Z}_{\epsilon,n}(\vec{x}):= -\ii\epsilon^2\frac{\nabla\psi_{\epsilon,n}(\vec{x})}{\psi_{\epsilon,n}(\vec{x})}= \epsilon^2\nabla\left(S_{\epsilon,n}-\ii R_{\epsilon,n}\right),
\end{equation}
it follows that,
 \begin{equation}\label{Zeqn}
 -\hf[\ii \epsilon^2]\nabla\cdot \vec{Z}_{\epsilon,n}(\vec{x}) + \hf \vec{Z}_{\epsilon,n}^2(\vec{x}) -\frac{\mu}{|\vec{x}|} = E_n,
 \end{equation}
 where in cartesians $\vec{Z}_{\epsilon,n}$ can be written,
 \begin{equation}\label{Z}
 \vec{Z}_{\epsilon,n}(\vec{x}) = \frac{\ii\mu}{\lambda }\left(1
  - \frac{\mathcal{L}'_{n-1}(n\nu)}{\mathcal{L}_{n-1}(n\nu)}\right)
   \frac{ \vec{x}}{|\vec{x}|}+ \frac{\mu}{\lambda e }\frac{\mathcal{L}'_{n-1}(n\nu)}{\mathcal{L}_{n-1}(n\nu)}
  \left(\ii,-\sqrt{1-e^2},0
 \right),
 \end{equation}
 and $\nu$ is defined in equation (\ref{nu}).

Combining equations (\ref{Zdefn}) and (\ref{Z}), the  drift for the diffusion $\vec{X}^{\epsilon,n}$ is,
 \begin{equation}\label{driftn}
 \vec{b}_{\epsilon,n}(\vec{x}) = \epsilon^2 \nabla(R_{\epsilon,n}+S_{\epsilon,n})=  \Re (\vec{Z}_{\epsilon,n}(\vec{x})) - \Im( \vec{Z}_{\epsilon,n}(\vec{x})),
 \end{equation}
where $\Re$ denotes the real part and $\Im$ denotes the imaginary part.

Clearly, from equation (\ref{Z}), the drift field $\vec{b}_{\epsilon,n}$ will have curves of singularities in the plane $y=0$ corresponding to the nodal curves of the wave function $\psi_{\epsilon,n}$ shown in Figure \ref{nodal} and also to the point $|\vec{x}|=0$. These singularities can be seen in Figure \ref{exactfig1} which shows the $x$ component of the drift in the plane $z=0$. The  line of singularities corresponds to the intersection of the nodal curves with the plane $z=0$. As discussed in \cite{MR882809}, at each of these nodes the drift is infinitely repulsive with the drift components becoming plus or minus infinity depending on the direction of approach. This means that the nodes will be  effectively inaccessible to the diffusion process. As $n$ increases the number and density of the nodes will also increase.

The invariant measure for the wave function $\psi_{\epsilon,n}$ in the plane $z=0$ is shown in Figure \ref{exactfig2} together with simulations of $\vec{X}^{\epsilon,n}$. The invariant measure has a clear peak on an ellipse in the plane $z=0$. Away from this plane it tends towards zero.  

\begin{figure}
\centering
\resizebox{120mm}{!}{\includegraphics{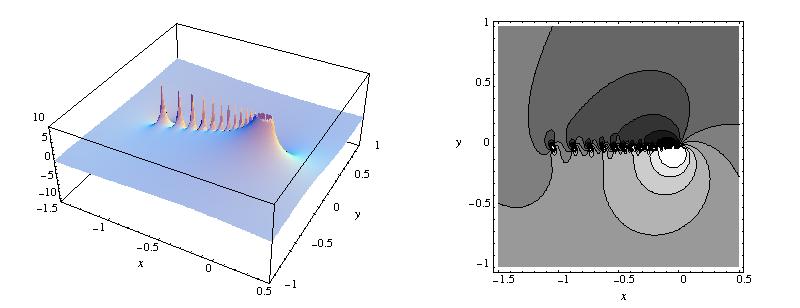}}
\caption{The $x$ component of the drift $\vec{b}_{\epsilon,n}$ in the plane $z=0$ for $n=20$ and $e=0.5$ shown as a surface plot and a contour plot (black =  $-\infty$, white = $+\infty$). }
\label{exactfig1}
\end{figure}
 \begin{figure}
\centering
\resizebox{120mm}{!}{\includegraphics{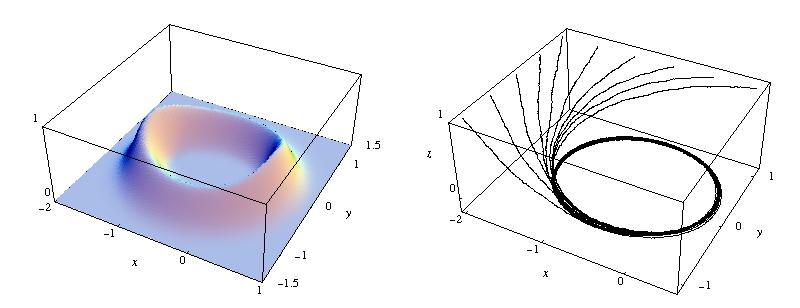}}
\caption{The invariant measure in the plane $z=0$ for the wave function $\psi_{\epsilon,n}$ and simulations of the process $\vec{X}^{\epsilon,n}$ for $n=20$ and $e=0.5$. }
\label{exactfig2}
\end{figure}

%%%%%%%%%%%%%%%%%%%%%%%%%%%%%%%%%%%%%%%%%%%%%%%%%%%%%%%%%%%%%%%%%%%%%%%
%%%%%%%%%%%%%%%%%%%%%%%%%%%%%%%%%%%%%%%%%%%%%%%%%%%%%%%%%%%%%%%%%%%%%%%
\section{The limiting atomic elliptic state}

\subsection{The limiting wave function}
We now derive the Bohr correspondence limit of the wave function $\psi_{\epsilon,n}$ as $n\rightarrow\infty$ and $\epsilon\rightarrow 0$ with $\lambda = \epsilon^2 n$ fixed.
Recall the function $\vec{Z}_{\epsilon,n}$ from equation (\ref{Zdefn}).
We define the Bohr correspondence limit of $\vec{Z}_{\epsilon,n}$ as,
 \[\vec{Z}_{0,\infty}(\vec{x}): = \lim\limits_{\genfrac{}{}{0pt}{}{n\uparrow\infty, \epsilon\downarrow 0}{\lambda = n\epsilon^2}}\vec{Z}_{\epsilon,n}(\vec{x}),\]
 where the limit is taken with $\lambda$ as a fixed real number.

 In what follows the Laguerre polynomials $\mathcal{L}_n(x)$ are defined according to the conventions in \cite{MR0240343}.
 \begin{lemma}\label{laguerre_lemma}
 Let $\mathcal{L}_n(x)$ denote the $n^{th}$ Laguerre polynomial and $\lambda$ be a fixed real number. Then,
\[\lim\limits_{\genfrac{}{}{0pt}{}{n\uparrow\infty, \epsilon\downarrow 0}{\lambda = n\epsilon^2}} \frac{\mathcal{L}'_{n-1}(n\nu)}{\mathcal{L}_{n-1}(n\nu)} = \hf\left(1-\sqrt{1-\frac{4}{\nu}}\right).\]
 \end{lemma}
 \begin{proof}
For Laguerre polynomials,
 \[v \mathcal{L}_{n-1}'(v) = (n-1) \mathcal{L}_{n-1}(v) - (n-1) \mathcal{L}_{n-2}(v),\]
 and so,
  \[\frac{\mathcal{L}_{n-1}'(v)}{\mathcal{L}_{n-1}(v)} = \frac{n-1}{v}-\frac{n-1}{v}\frac{\mathcal{L}_{n-2}(v)}{\mathcal{L}_{n-1}(v)}.\]
Setting $v = n\nu$ gives,
 \[\lim\limits_{\genfrac{}{}{0pt}{}{n\uparrow\infty, \epsilon\downarrow 0}{\lambda = n\epsilon^2}} \frac{\mathcal{L}'_{n-1}(n\nu)}{\mathcal{L}_{n-1}(n\nu)}
 = \frac{1}{\nu}-\frac{1}{\nu}\lim\limits_{\genfrac{}{}{0pt}{}{n\uparrow\infty, \epsilon\downarrow 0}{\lambda = n\epsilon^2}} \frac{\mathcal{L}_{n-2}(n\nu)}{\mathcal{L}_{n-1}(n\nu)}.\]
 However, Laguerre polynomials also satisfy the recurrence relation,
 \[n \mathcal{L}_n(v) -(2n - 1-v)\mathcal{L}_{n-1}(v) + (n-1) \mathcal{L}_{n-2}(v)=0.\]
 Thus, if the limit,
 \[p = \lim\limits_{\genfrac{}{}{0pt}{}{n\uparrow\infty, \epsilon\downarrow 0}{\lambda = n\epsilon^2}}\frac{\mathcal{L}_{n-2}(n\nu)}{\mathcal{L}_{n-1}(n\nu)},\]
 exists and is non zero, then necessarily $p$ satisfies,
 \[\frac{1}{p} - (2-\nu) +p = 0,\]
 proving the lemma.
 \end{proof}
Applying Lemma \ref{laguerre_lemma} to $\vec{Z}_{\epsilon,n}$  gives, in cartesians,
\begin{equation}\label{Z0}
\vec{Z}_{0,\infty}(\vec{x})
= \frac{\ii\mu}{2\lambda }\left(
   1+\sqrt{1-\frac{4}{\nu}}\right)
    \frac{\vec{x}}{|\vec{x}|}+ \frac{\mu}{2\lambda e }\left(1-\sqrt{1-\frac{4}{\nu}}\right)
  \left(\ii,-\sqrt{1-e^2},0
 \right),\end{equation}
 where as expected from equation (\ref{Zeqn}),
 \[\hf \vec{Z}_{0,\infty}^2(\vec{x}) -\frac{\mu}{|\vec{x}|} = -\frac{\mu^2}{2\lambda^2}.\]

\begin{prop}
The Bohr correspondence limit of the wave function $\psi_{\epsilon,n}$ for the atomic elliptic state gives the formal limiting wave function,
 \begin{equation}\label{limiting}
 \psi_{\epsilon}(\vec{x}) = \nu^{\frac{\lambda}{\epsilon^2}}\left(1+\sqrt{1-\frac{4}{\nu}}\right)^{\frac{2\lambda}{\epsilon^2} }\exp\left(-\frac{\mu }{\lambda\epsilon^2}|\vec{x}| + \frac{\lambda \nu}{2\epsilon^2}\left(1-\sqrt{1-\frac{4}{\nu}}\right)\right).
 \end{equation}
 \end{prop}
 \begin{proof}
 For the vector $\vec{Z}_{0,\infty}$ the function $\psi_{\epsilon}$ is the formal wave function satisfying,
 \begin{equation}\label{limitZpsi}
 \vec{Z}_{0,\infty}(\vec{x}) = -\ii\epsilon^2\frac{\nabla \psi_{\epsilon}(\vec{x})}{\psi_{\epsilon}(\vec{x})}.
 \end{equation}
 This wave function corresponds to Bohr's limit but is only an approximate solution of the Schr\"{o}dinger equation.
 \end{proof}

 %%%%%%%%%%%%%%%%%%%%%%%%%%%%%%%%%%%%%%%%%%%%%%%%%%%%%%%%%%%%%%%%%%%%%%%
\subsection{The limiting Nelson diffusion process}
We now construct the limiting Nelson diffusion $\vec{X}^{\epsilon}$ corresponding to the  limiting wave function $\psi_{\epsilon}$.
As in equation (\ref{Zdefn}),
\[\vec{Z}_{0,\infty}(\vec{x}) = -\ii\epsilon^2\frac{\nabla \psi_{\epsilon}(\vec{x})}{\psi_{\epsilon}(\vec{x})} = \epsilon^2\nabla\left(S_{\epsilon}-\ii R_{\epsilon}\right),\]
where $\psi_{\epsilon}(\vec{x}) = \exp(R_{\epsilon}+\ii S_{\epsilon})$. The subscripts emphasise the $\epsilon$ dependence in $R_{\epsilon}$ and $S_{\epsilon}$ arising from the $\epsilon^2$.

\begin{prop}\label{prop_Nelson_3d}
The limiting Nelson diffusion process $\vec{X}^{\epsilon}$ corresponding to the limiting wave function $\psi_{\epsilon}$ satisfies the It\^{o} equation,
\[
 \di \vec{X}^{\epsilon}(s)  = \vec{b}(\vec{X}^{\epsilon}(s))\di s +\epsilon \di \vec{B}(s),
 \]
 where $\vec{b}(\vec{x}) = (b_x,b_y,b_z)$ in cartesian coordinates with,
\begin{subequations}\label{bb}
 \begin{eqnarray}
 b_x & = & \frac{\mu}{2\lambda}\left\{(\alpha+\beta-1)\frac{1}{e} - (\alpha+\beta+1)\frac{x}{|\vec{x}|}\right\},\label{bx} \\
 b_y & = & \frac{\mu}{2\lambda}\left\{(\alpha-\beta-1)\frac{\sqrt{1-e^2}}{e} - (\alpha+\beta+1)\frac{y}{|\vec{x}|}\right\}, \label{by}\\
 b_z & = & -\frac{\mu}{2\lambda}(\alpha+\beta+1)\frac{z}{|\vec{x}|},\label{bz}
 \end{eqnarray}
 \end{subequations}
  and,
 \begin{subequations}\label{alphabeta}
 \begin{eqnarray}
 \alpha & = & \left(\hf\sqrt{\frac{\left(e|\vec{x}|-x-4\lambda^2 e/\mu\right)^2 +(1-e^2)y^2}{(e |\vec{x}|-x)^2+(1-e^2)y^2}}+\right.\nonumber\\
 &  &\qquad\left.+\hf\frac{\left(e|\vec{x}|-x-2\lambda^2e/\mu\right)^2 +(1-e^2)y^2
 -4\lambda^4 e^2/\mu^{2}}{(e |\vec{x}|-x)^2+(1-e^2)y^2}\right)^{\hf}_,\label{alpha}\\
 \beta & = & \frac{- 2\lambda^2 e \sqrt{1-e^2} y}{\mu\left((e|\vec{x}|-x)^2+(1-e^2)y^2\right)\alpha}.\label{beta}
 \end{eqnarray}
\end{subequations}
\end{prop}
\begin{proof}
The drift term in the limiting Nelson diffusion process $\vec{X}^{\epsilon}$ corresponding to the limiting wave function $\psi_{\epsilon}$ is,
\[
 \vec{b}(\vec{x}) = \epsilon^2 \nabla(R_{\epsilon }+S_{\epsilon})=  \Re (\vec{Z}_{0,\infty}(\vec{x})) - \Im( \vec{Z}_{0,\infty}(\vec{x})).
 \]
 A simple calculation gives (for $\nu$ as in equation (\ref{nu})),
 \[\sqrt{1-\frac{4}{\nu}} = \alpha+i\beta,\]
 and the result follows.
\end{proof}
 Moreover, we can also find the functions $R_{\epsilon}$ and $S_{\epsilon}$ explicitly,
  \begin{eqnarray*}
R_{\epsilon} & = & \frac{\lambda}{2\epsilon^2}\left(\ln\left(\tilde{\alpha}^2+\tilde{\beta}^2\right)
+2\ln\left((1+\alpha)^2+\beta^2\right)+(1-\alpha)\tilde{\alpha}+ \beta\tilde{\beta}\right)-\frac{\mu |\bm{x}|}{\lambda\epsilon^2},\\
S_{\epsilon} & = & \frac{\lambda}{\epsilon^2}\left(\arg\left(\tilde{\alpha}+\ii \tilde{\beta}\right)+
2 \arg\left(1+\alpha+\ii \beta\right)+\hf\tilde{\beta}(1-\alpha) - \hf\beta\tilde{\alpha}\right),
 \end{eqnarray*}
 where,
 \[\tilde{\alpha}= \frac{\mu}{\lambda^2}\left(|\vec{x}| -\frac{x}{e} \right),\qquad   \tilde{\beta} = -\frac{\mu y \sqrt{1-e^2}}{\lambda^2e}.\]

Our main object of study for the rest of this paper is the limiting Nelson diffusion process $\vec{X}^{\epsilon}$
in Proposition \ref{prop_Nelson_3d}.
Figure \ref{3dsims} shows several simulations of sample paths for $\vec{X}^{\epsilon}$. These paths all converge to the Kepler ellipse.
 We hope to recover  Keplerian motion on the Kepler ellipse from the underlying deterministic dynamical system $\vec{X}^0$ found from $\vec{X}^{\epsilon}$ in the limit as $\epsilon\rightarrow 0$.
 \begin{figure}
\centering
\resizebox{120mm}{!}{\includegraphics{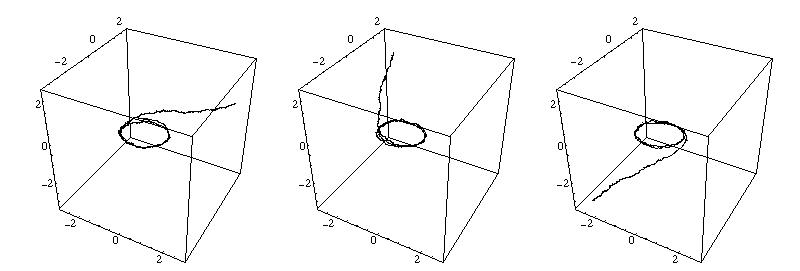}}
\caption{Simulations of the limiting diffusion $\vec{X}^{\epsilon}$ with $e=0.5$.}
\label{3dsims}
\end{figure}

\begin{remark}
The process $\vec{X}^{\epsilon}$ is the formal limit of the process $\vec{X}^{\epsilon,n}$ found by letting $n\rightarrow \infty$ whilst $\epsilon\rightarrow 0 $ with $\lambda = n\epsilon^2$ fixed, in the drift term.
The process $\vec{X}^{\epsilon,n}$ before taking this limit  satisfies the Nelson-Newton law,
 \[\hf(D_+D_-+D_-D_+)\vec{X}^{\epsilon,n}(t) = -\mu\frac{\vec{X}^{\epsilon,n}(t)}{|\vec{X}^{\epsilon,n}(t)|^3}.\]
\end{remark}

 %%%%%%%%%%%%%%%%%%%%%%%%%%%%%%%%%%%%%%%%%%%%%%%%%%%%%%%%%%%%%%%%%%%%%%%
 \subsection{The drift singularity}
 Considering equations (\ref{alphabeta}) in Proposition \ref{prop_Nelson_3d}, it is clear that $\alpha$  and $\beta$ will both have singularities at the point $|\vec{x}|=0$. However, $\beta$ will also have a singularity along the surface $\alpha=0$. Therefore,   the drift could be singular at any of these points.

 The former is easy to analyse using spherical polar coordinates,
 \[x = r \cos\theta\sin\phi,\quad y = r \sin\theta\sin\phi, \quad z = r \cos\phi.\]
 Then,
 \begin{eqnarray*}
 \alpha & \sim & C(\theta,\phi) r^{-\hf}, \\
 \beta & = & \frac{- 2 \lambda^2 e \sqrt{1-e^2}  \sin\theta\sin\phi}{\left((e- \cos\theta\sin\phi)^2+(1-e^2) \sin^2\theta\sin^2\phi\right)\mu r\alpha}
 \sim - K(\theta,\phi) \sin\theta\sin\phi r^{-\hf},
 \end{eqnarray*}
 as $r\sim 0$, where $C,K$ are some positive functions independent of $r$. One would normally expect the probability of the diffusion $\vec{X}^{\epsilon}$ hitting this point singularity in finite time in two or more dimensions to be zero.

We now consider the behaviour on the surface $\alpha =0$.  Taking equation (\ref{alpha}), we define $\alpha_1$ and $\alpha_2$ so that,
\begin{equation}\label{alpha1}\alpha = \sqrt{\hf\left(\sqrt{\alpha_1} + \alpha_2\right)}.
\end{equation}
Clearly, a necessary condition for $\alpha = 0$ is $\alpha_1 -\alpha_2^2 = 0,$
and so working in cylindrical polar coordinates,
\[x = \tilde{r} \cos\theta,\quad y = \tilde{r}\sin\theta,\quad z = z,\]
we have,
\[\alpha_1 - \alpha_2^2 = \frac{64 \lambda^4 e^2(1-e)  (e+1) \tilde{r}^2 \sin ^2\theta }{\mu^2\left(\left(e^2\cos 2 \theta +e^2+2\right) \tilde{r}^2-4 e\tilde{r} \sqrt{\tilde{r}^2+z^2} \cos \theta  +2 e^2 z^2\right)^2},
\]
and so $\alpha$ is only zero on $y=0$. Moreover, returning to cartesians,
\begin{equation}\label{square1}
\alpha|_{y=0} =\sqrt{\hf\left(\frac{4\lambda^2 e /\mu-e \sqrt{x^2+z^2} +x}{x-e \sqrt{x^2+z^2}}+\sqrt{\frac{\left(4  \lambda^2 e/\mu-e\sqrt{x^2+z^2} +x\right)^2}{\left(x-e \sqrt{x^2+z^2}\right)^2}}\right)},
\end{equation}
and so we can conclude,
\[\alpha = 0 \quad\Leftrightarrow \quad \vec{x}=(x,y,z) \in \Sigma,\]
where $\Sigma$ is the set,
\begin{equation}\label{singu}
\Sigma = \left\{(x,0,z):\, \frac{-e \left(4\lambda^2/\mu-\sqrt{\left(16 \lambda^4/\mu^2-z^2\right) e^2+z^2}\right)}{1-e^2}<x< \sqrt{\frac{e^2z^2}{1-e^2}}\right\}.
\end{equation}
The set $\Sigma$ is shown in Figure \ref{3dsing} together with the Kepler ellipse. This singularity is what one would expect if we consider $\Sigma$ as the correspondence limit of the nodal surfaces of the wave function $\psi_{\epsilon,n}$ (see Figure \ref{nodal}).
It can be shown that this singularity leads to a finite jump discontinuity in the drift.

\begin{figure}
\centering
\resizebox{120mm}{!}{\includegraphics{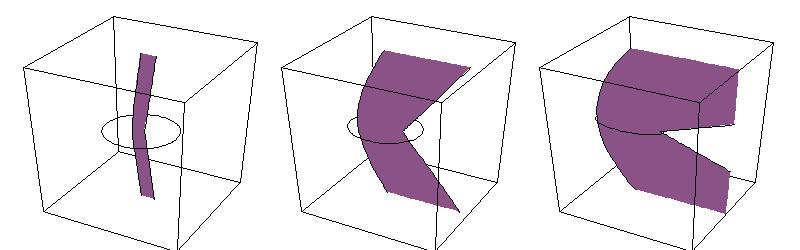}}
\caption{The drift singularity (shaded)  with the Kepler ellipse for $e=0.1,$ $0.5$ and $ 0.9$.}
\label{3dsing}
\end{figure}

 %%%%%%%%%%%%%%%%%%%%%%%%%%%%%%%%%%%%%%%%%%%%%%%%%%%%%%%%%%%%%%%%%%%%%%%
\subsection{Restriction to two dimensions}
For the remainder of this paper we shall restrict our wave functions $\psi_{\epsilon,n}$ and $\psi_{\epsilon}$ and their related diffusion processes $\vec{X}^{\epsilon,n}$ and $\vec{X}^{\epsilon}$ to the putative plane of motion $z=0$.  We will consider the full three dimensional problem in  a future paper.

\begin{prop}
If $\psi_{\epsilon}(x,y,z)$ is the three dimensional limiting wave function, then its restriction to the plane $z=0$, is given by,
\begin{eqnarray*}
\psi_{\epsilon}(x,y) &=&  \nu^{\frac{\lambda}{\epsilon^2}}\left(1+\sqrt{1-\frac{4}{\nu}}\right)^{\frac{2\lambda}{\epsilon^2} }\times\\
 & & \quad \!\! \exp\left(-\frac{\mu }{\lambda\epsilon^2}
 \sqrt{x^2+y^2} + \frac{\lambda \nu}{2\epsilon^2}\left(1-\sqrt{1-\frac{4}{\nu}}\right)\right),
\end{eqnarray*}
where
\[\nu = \frac{\mu}{\lambda^2}\left(\sqrt{x^2+y^2} - \frac{x}{e} - \frac{\ii y\sqrt{1-e^2}}{e}\right).
\]
\end{prop}
\begin{proof}
We state without proof that the correspondence limit in three dimensions satisfies,
 \[\psi_{\epsilon}^{\mathrm{3-dim}}|_{z=0} = \psi_{\epsilon}^{\mathrm{2-dim}},\]
giving the result.
\end{proof}
\begin{prop}\label{prop4}
The limiting Nelson diffusion process $\vec{X}^{\epsilon}=(X_x^{\epsilon},X_y^{\epsilon})$ restricted to the plane $z=0$ satisfies,
\[
 \di \vec{X}^{\epsilon}(s)  = \vec{b}(\vec{X}^{\epsilon}(s))\di s +\epsilon \di \vec{B}(s),
 \]
where $\vec{b}(\vec{x}) = (b_x,b_y)$ with,
\begin{subequations}\label{bb2d}
 \begin{eqnarray}
 b_x & = & \frac{\mu}{2\lambda}\left\{(\alpha+\beta-1)\frac{1}{e} - (\alpha+\beta+1)\frac{x}{\sqrt{x^2+y^2}}\right\},\label{bx2d} \\
 b_y & = & \frac{\mu}{2\lambda}\left\{(\alpha-\beta-1)\frac{\sqrt{1-e^2}}{e} - (\alpha+\beta+1)\frac{y}{\sqrt{x^2+y^2}}\right\}, \label{by2d}
 \end{eqnarray}
 \end{subequations}
 and,
 \begin{subequations}\label{alphabeta2d}
 \begin{eqnarray}
 \alpha & = & \left(\hf\sqrt{\frac{\left(e\sqrt{x^2+y^2}-x-4\lambda^2 e/\mu\right)^2 +(1-e^2)y^2}{\left(e \sqrt{x^2+y^2}-x\right)^2+(1-e^2)y^2}}+\right.\nonumber\\
 &  &\quad\left.+\hf\frac{\left(e\sqrt{x^2+y^2}-x-2\lambda^2e/\mu\right)^2 +(1-e^2)y^2
 -4\lambda^4 e^2/\mu^{2}}{\left(e \sqrt{x^2+y^2}-x\right)^2+(1-e^2)y^2}\right)^{\hf}_,\label{alpha2d}\\
 \beta & = & \frac{- 2\lambda^2 e \sqrt{1-e^2} y}{\mu\left(\left(e\sqrt{x^2+y^2}-x\right)^2+(1-e^2)y^2\right)\alpha}.\label{beta2d}
 \end{eqnarray}
\end{subequations}
\end{prop}
\begin{proof}
This follows  from Proposition \ref{prop_Nelson_3d}.
\end{proof}
Simulations of the restricted limiting Nelson diffusion process are shown in Figure \ref{2dsims}.
 \begin{figure}
\centering
\resizebox{120mm}{!}{\includegraphics{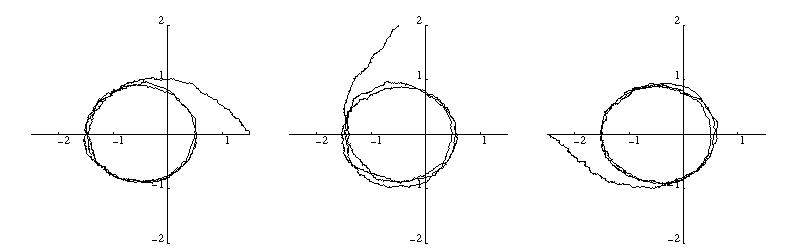}}
\caption{Simulations of the 2 dimensional diffusion $\vec{X}^{\epsilon}$ for $e=0.5$.}
\label{2dsims}
\end{figure}

%%%%%%%%%%%%%%%%%%%%%%%%%%%%%%%%%%%%%%%%%%%%%%%%%%
%%%%%%%%%%%%%%%%%%%%%%%%%%%%%%%%%%%%%%%%%%%%%%%%%%%
 \section{The Keplerian elliptic coordinate system}
In this section we will define a new two dimensional coordinate system to simplify the limiting diffusion $\vec{X}^{\epsilon}$.
 Recall from Definition \ref{Keplerellipse} that the Kepler ellipse is the ellipse with eccentricity $e$ and semi-major axis $a = \lambda^2/\mu$ with one focus at the origin and the other at $(-2ae,0)$.

   In the plane $z=0$, the singularity $\Sigma$ defined in equation (\ref{singu}) reduces to,
\begin{equation}\label{singu2d}
\Sigma = \left\{(x,0) : \quad  -\frac{4a e}{1+e}<x<0\right\}.
\end{equation}
We want to find a coordinate system which will simplify the complex expressions  $\alpha$ and $\beta$ in Proposition \ref{prop4}. From equations (\ref{alpha1}) and  (\ref{square1}) it is apparent that on the singularity $\alpha_1$  becomes a perfect square. A simple calculation shows that the same also happens on the Kepler ellipse.

In fact if we consider this square root term in polar coordinates,
\[\alpha_1 = \frac{8 \left(2 \lambda ^4-r \lambda ^2 \mu \right) e^2+r \mu  \cos \theta  \left(8 \lambda ^2-2 r \mu +e r \mu  \cos \theta \right) e+r^2 \mu ^2}{r^2 \mu ^2 (e \cos \theta -1)^2},\]
and evaluate this on an arbitrary ellipse of semimajor axis $\gamma$ and eccentricity $c$, that is where $r=(1-c^2) \gamma /(c \cos \theta +1)$,
then we have,
\[\alpha_1\!=\! \frac{16 e^2 c ^2 \lambda ^4-8 e c \mu  \lambda ^2 \left(\gamma ( c ^2-1)+r(e  c +1) \right) +\mu^2\!\left(e \gamma ( c ^2-1)+r( c +e)\right)^2
 }{\left (e \gamma ( c ^2-1)+r( c +e)\right) ^2 \mu ^2}.\]
The numerator is a quadratic in $r$ and so forms a perfect square when its discriminant is zero. This only occurs when $\gamma =2 a e/(c+e).$
Therefore,  there is an infinite family of ellipses, including both the Kepler ellipse and  the singularity (as a degenerate ellipse), on which $\alpha_1$ forms a perfect square.

\begin{defn} The family of ellipses $\mathcal{E}_c$ are non-confocal ellipses with eccentricity $|c|$, foci at $(0,0)$ and $(-4 a e c/(1+e),0)$ and with semi major axis ${2 a e}/({e+c})$ where $-e\le c\le 1$.
\end{defn}
Using the ellipses $\mathcal{E}_u$ we introduce $(u,v)$ coordinates.
\begin{defn}
The Keplerian elliptic coordinates $(u,v)$ in the plane $(x,y)$ are defined by,
\begin{equation}\label{coords}
 x = \frac{2 a e (\cos v - u)}{e+u},\quad  y = \frac{2 a e \sqrt{1-u^2} \sin (v)}{e+u},
  \end{equation}
where $-e<  u \le 1$ and $0\le v<2 \pi $.
\end{defn}
 These are non-orthogonal coordinates. The ellipse $\mathcal{E}_c$ corresponds to the coordinate curve  $u = c$ for some constant $c\in(-e,1]$.
The Kepler ellipse  is the curve $u = e$ (i.e. $\mathcal{E}_e$) and the singularity is the degenerate ellipse given by $u=1$ (i.e. $\mathcal{E}_1$). The ellipse at infinity is given by $u = -e$ (i.e. $\mathcal{E}_{-e}$).
  The curves $v = k$ for some constant $k$, are similar to hyperbolas.
%   satisfying,
% \[ \left(x+\frac{2 a e u}{e+u}\right)^2- \frac{y^2}{ \left(u^2-1\right)}=\frac{4 a^2 e^2}{(e+u)^2}.\]
The coordinate curves are shown in Figure \ref{coordpic} for several values of $e$. It is important to note the bunching of the  curves of constant $u$ which occurs at the left hand end of the singularity for large eccentricities.
\begin{figure}
\centering
\resizebox{120mm}{!}{\includegraphics{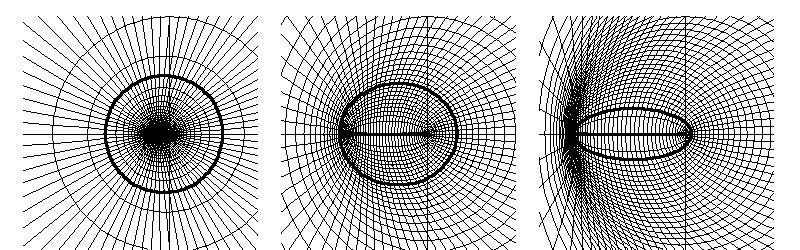}}
\caption{The Keplerian elliptic coordinate system for $e=0.1$, $0.5$ and $0.9$ with the singularity ($u=1$) and Kepler ellipse ($u=e$)  in bold.}
\label{coordpic}
\end{figure}

The usefulness of Keplerian elliptic coordinates is immediately  apparent as equations (\ref{alphabeta2d}) simplify to give,
 \begin{equation}\label{alphabetauv}\alpha = \frac{\sqrt{(1-u^2)(1-e^2)}}{1+eu-(e+u)\cos v},\qquad \beta = -\frac{(e+u)\sin v}{1+ eu-(e+u)\cos v},
 \end{equation}
 and the drift from equations (\ref{bb2d}) becomes,
 \begin{eqnarray*}
b_x \!\!& = &\! \!\frac{\mu}{2\lambda}\left(\frac{e u-(e-u) \cos v-(e+u) \sin v+\sqrt{\left(1-e^2\right) \left(1-u^2\right)}-1}{e
   (1-u \cos v)}\right),\\
     b_y \!\!& = &\! \!\frac{\mu}{2\lambda}\left( \frac{\sqrt{1-u^2}\left(e \cos v-e \sin v+1\right)
     +\sqrt{1-e^2}\left(u\cos v +u\sin v-1\right)}{e(1- u \cos v)}\right).
 \end{eqnarray*}
The singularity $\Sigma$ becomes the line $u=1$ with approaches from above (i.e. $y>0$) corresponding to $0<v<\pi$ and approaches from below (i.e. $y<0$) as $\pi<v<2\pi$. In this manner the singularity is opened out onto the boundary of our new coordinate space which forms a cylinder. It is also apparent that there is a finite jump discontinuity in $\beta$ across $\Sigma$ as,
\[\beta\big|_{u=1,v=v} = -\frac{\sin v}{1-\cos v},\qquad \beta\big|_{u=1,v=2\pi-v} =\frac{\sin v}{1-\cos v}.\]

We can also use our new coordinates to rewrite the diffusion process $\vec{X}^{\epsilon}$ in $(u,v)$ space.

\begin{prop}\label{itocorrection}
The diffusion $\vec{X}^{\epsilon}$ can be written in terms of the Keplerian elliptic cooridinates $(u,v)$ as,
\begin{eqnarray*}
\di X_u^{\epsilon} &=& h(X_u^{\epsilon},X_v^{\epsilon})\left\{b_u({X}_u^{\epsilon},{X}_v^{\epsilon}) \di t -\epsilon \vec{N}(X_u^{\epsilon},X_v^{\epsilon})\cdot (\di B_1,\di B_2)\right\},\\
 \di X_v^{\epsilon} &=&h(X_u^{\epsilon},X_v^{\epsilon})\left\{ b_v({X}_u^{\epsilon},{X}_v^{\epsilon}) \di t -\epsilon \vec{M}(X_u^{\epsilon},X_v^{\epsilon})\cdot (\di B_1,\di B_2)\right\},
 \end{eqnarray*}
 where
\begin{eqnarray*}
\vec{N}(u,v) & = &\left( (e+u)(1-u^2)\cos v, (e+u)\sqrt{1-u^2} \sin v\right),\\
\vec{M}(u,v) & = & \left((1+eu)\sin v, -\sqrt{1-u^2}(e+\cos v)\right),
\end{eqnarray*}
with,
\[h(u,v) = \frac{(e+u)}{2 a e (1-u \cos v) (e u+(e+u) \cos v+1)},\]
and,
\begin{eqnarray*}
b_u
  &= & \epsilon^2 I_u(u,v)-\frac{\mu}{2 e \lambda}(e+u)\sqrt{1-u^2}\\
   & &\quad \left(
\sqrt{1-e^2} \left(u+\cos v-\sin v\right)
      -\sqrt{1-u^2} \left(e+\cos v-\sin v\right)\right)\nonumber\\
   b_v &= & \epsilon^2 I_v(u,v)-\frac{\mu}{2 e \lambda}\bigg( \sqrt{\left(1-e^2\right) \left(1-u^2\right)} (e+\cos v+\sin v)\nonumber\\
   & & \quad -u(1+e^2)-2e-\left(e^2+2 u e+1\right) \cos v-\left(1-e^2\right) \sin  v\bigg),
\end{eqnarray*}
with the It\^{o} correction terms,
\begin{eqnarray*}
I_u(u,v)& = &
\frac{-(e+u)^2}{4 a e (1+eu+(e+u)\cos v)^2}
\bigg((e+u)^2\left( \left(2 u^2-1\right) \cos^2  v+1\right) +\\
&&\quad
+ 2u \left((e+u)^2-(1-u^2)(1+eu)\right) \cos v -(1-u^2)( 1-e^2)
\bigg),\\
I_v(u,v)  & = & \frac{(e+u)\sin v}{4 a  (1+eu+(e+u)\cos v)^2}
\bigg(2(u+e)^2-(1+eu)^2\\
& & \quad+(e+u) (e u+1) \cos v\bigg).
\end{eqnarray*}
\end{prop}
\begin{proof}
This follows from applying It\^{o}'s formula to equations (\ref{coords}) which define the Keplerian elliptic coordinates.
\end{proof}

%%%%%%%%%%%%%%%%%%%%%%%%%%%%%%%%%%%%%%%%%%%%%%%%%%%%%%%%%%%%%%%%%%%%%%%
%%%%%%%%%%%%%%%%%%%%%%%%%%%%%%%%%%%%%%%%%%%%%%%%%%%%%%%%%%%%%%%%%%%%%%%
\section{The limiting Nelson diffusion $\vec{X}^{\epsilon}$}
\subsection{The invariant measure and the infinite time limit}
We want to show that the (time dependent) density $\rho_{\epsilon}(\vec{x},t)$ for the diffusion $\vec{X}^{\epsilon}$ with any given initial distribution $\rho_{\epsilon}^0(\vec{x})$ will converge in the infinite time limit to a density concentrated on the Kepler ellipse. We begin by showing that the invariant measure  is this  physically correct density in the limit as $\epsilon\rightarrow 0$.

The invariant measure for the Nelson diffusion process $\vec{X}^{\epsilon,n}$ corresponding to the atomic elliptic state wave function $\psi_{\epsilon,n} = \exp(R_{\epsilon,n} + iS_{\epsilon,n})$ is given by,
\[\rho_{\epsilon,n}^{\infty}(\vec{x}):= \frac{\psi_{\epsilon,n}\psi^*_{\epsilon,n}}{||\psi_{\epsilon,n}||^2}= \frac{\exp (2R_{\epsilon,n}(\vec{x}))}{\iint \exp (2R_{\epsilon,n}(\vec{x}))\di x\di y}.\]
This density is shown in Figure \ref{exactfig2} where $n=20$ and $\epsilon^2 = \lambda/n$.
For the limiting wave function $\psi_{\epsilon} = \exp(R_{\epsilon}+iS_{\epsilon})$ the invariant measure is $\rho_{\epsilon}^{\infty}$, where
 \begin{equation}
 \label{invariantdensity}
 \rho_{\epsilon}^{\infty}(\vec{x}): = \frac{\exp (2R_{\epsilon}(\vec{x}))}{\iint \exp (2R_{\epsilon}(\vec{x}))\di x\di y}.
 \end{equation}

\begin{theorem}\label{invariantthm}
The invariant density $\rho_{\epsilon}^{\infty}$ coming from the Bohr correspondence limit of the atomic elliptic state has the Kepler ellipse $\mathcal{E}_e$ as a manifold of maxima on which it attains a constant global maximum. This maximum is sharply peaked on $\mathcal{E}_e$ as $\epsilon\rightarrow 0$.
\end{theorem}
\begin{proof}
We can write $R_{\epsilon}$ in terms of Keplerian elliptic coordinates $(u,v)$ giving,
\begin{eqnarray*}
\exp(2 R_{\epsilon})\! &  = & \! 16^{\frac{\lambda
   }{\epsilon ^2}}
\left(\frac{1+e
   u+\sqrt{\left(1-e^2\right)
   \left(1-u^2\right)}}{e+u}\right)^{\frac{2 \lambda
   }{\epsilon ^2}}\times
\\
& &  \exp\left(\frac{2 a \mu
   \left(u-e+\left(e u -1+\sqrt{\left(1-e^2\right)
   \left(1-u^2\right)}\right) \cos v\right)}{(e+u)
   \epsilon ^2 \lambda }\right).
   \end{eqnarray*}
   On $\mathcal{E}_e$ this is constant,
   \[\exp(2 R_{\epsilon}(e,v))=\left(\frac{16}{e^2}\right)^{\frac{\lambda }{\epsilon ^2}}\quad\Rightarrow\quad \left.\frac{\partial}{\partial v} \exp(2 R_{\epsilon}(u,v))\right|_{u=e} = 0.\]
   Moreover, a simple calculation produces,
\[\left.\frac{\partial}{\partial u} \exp(2 R_{\epsilon}(u,v))\right|_{u=e} = 0,\]
so that $\mathcal{E}_e$ is  a manifold of stationary points. A further calculation shows that this is the unique maximum for this function. Clearly, taking into account the normalisation factor in $\rho_{\epsilon}^{\infty}$, the invariant density will become sharply peaked as $\epsilon\rightarrow 0$, as shown in Figure \ref{rhopic}.
\end{proof}
 \begin{figure}
 \centering
\resizebox{130mm}{!}{\includegraphics{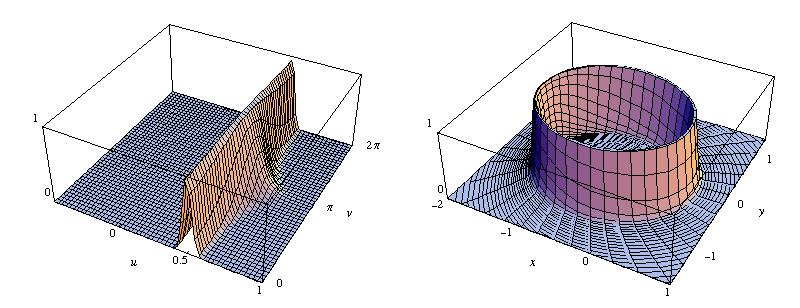}}
\caption{$\rho^{\infty}_{\epsilon}$ in $(u,v)$ and $(x,y)$ space for $e=0.5$ and $\epsilon=0.1$.}
\label{rhopic}
\end{figure}

We now want to show that $\rho_{\epsilon}(\vec{x},t)$ will converge to $\rho_{\epsilon}^{\infty}(\vec{x})$ in the infinite time limit. For this we introduce a similarity transform first developed in \cite{MR1189398} and applicable to any Nelson diffusion
$\vec{X}$ with diffusion constant $\epsilon^2/2$. The generator $\mathcal{G}$ of the diffusion $\vec{X}$ is the operator,
\[\mathcal{G} = \hf\epsilon^2\Delta +\vec{b}\cdot\nabla,\]
where $\vec{b} = \epsilon^2\nabla(R+S)$ for suitable real functions $R$ and $S$ such that the corresponding wave function can be written $\psi = \exp(R+\ii S)$.

\begin{lemma}\label{transform}
For the generator $\mathcal{G}$ of a Nelson diffusion process $\vec{X}$ with diffusion constant $\epsilon^2/2$,
\[\mathcal{G} = \exp(-(R+S)) \left(-\tilde{H}/\epsilon^2\right)\exp(R+S),\]
where $\tilde{H}$ is the formal Hamiltonian,
\[\tilde{H} = \hf\left(-\epsilon^4\Delta +\vec{b}^2 +\epsilon^2\nabla\cdot\vec{b}\right),\]
and $\vec{b} = \epsilon^2\nabla(R+S)$. Moreover,  for $\tilde{\psi} = \exp(R-S)$,
\[\tilde{H}\tilde{\psi} = 0,\]
so that $\tilde{\psi}$ is the formal ground state for $\tilde{H}$.
\end{lemma}
\begin{proof}
This follows from a calculation using the identity,
\[\Delta(fg) = f\Delta g + g\Delta f +2\nabla f\cdot\nabla g,\]
together with the real valued nature of $E,V,R$ and $S$.
\end{proof}

Recall that the atomic elliptic state wave function $\psi_{\epsilon,n}$ associated with the diffusion $\vec{X}^{\epsilon,n}$
  is an  exact solution for the Schr\"{o}dinger equation. We begin by showing how for the diffusion $\vec{X}^{\epsilon,n}$ with any given initial distribution $\rho_{\epsilon,n}^0(\vec{x})$ the (time dependent) density $\rho_{\epsilon,n}(\vec{x},t)$ should converge in the infinite time limit to the invariant measure $\rho^{\infty}_{\epsilon,n}(\vec{x})$.

The  density $\rho_{\epsilon,n}(\vec{x},t)$ satisfies the forward Kolmogorov equation (\ref{forward}),
\[\frac{\partial\rho_{\epsilon,n}(\vec{x},t)}{\partial t} = \mathcal{G}^*\rho_{\epsilon,n}(\vec{x},t)
, \quad \rho_{\epsilon,n}(\vec{x},0) = \rho^0_{\epsilon,n}(\vec{x}).\]
Therefore,  using Lemma \ref{transform}, we can write,
\[\rho_{\epsilon,n}(\vec{x},t) =\left( \exp(R_{\epsilon,n}+S_{\epsilon,n})\exp\left(-t\tilde{H}_{\epsilon,n}/\epsilon^2\right)
\exp(-(R_{\epsilon,n}+S_{\epsilon,n}))\rho^0_{\epsilon,n}\right)(\vec{x}),\]
where  $\tilde{H}_{\epsilon,n} = \hf(-\epsilon^4\Delta+\vec{b}_{\epsilon,n}^2+\epsilon^2\nabla\cdot\vec{b}_{\epsilon,n})$ is a formal Hamiltonian associated with the diffusion $\vec{X}_{\epsilon,n}$. For fixed $n$, the work of Blanchard and Golin \cite{MR882809} guarantees that with probability one the diffusion $\vec{X}^{\epsilon,n}$ cannot reach the nodal surfaces in a finite time because they are infinitely repulsive.  Therefore, for this case we do not need to worry about specifying a boundary condition at these points.

 However, if we are to consider the case $n\rightarrow\infty$ which gives the limiting diffusion $\vec{X}^{\epsilon}$, we need to specify a boundary condition for $\rho_{\epsilon}(\vec{x},t)$ across the limiting singularity $\Sigma$. If we assume that the forward Kolmogorov equation is valid on $\Sigma$ then,
\[\hf\mathrm{disc}_{\Sigma}\left(\epsilon^2\frac{\partial}{\partial \vec{n}}\ln \rho_{\epsilon}(\vec{x},t)\right) = \mathrm{disc}_{\Sigma}(\vec{b}\cdot\vec{n}),\]
where $\vec{n}$ is the unit normal to the singularity $\Sigma$ and  $\mathrm{disc}_{\Sigma}$ is the discontinuity across $\Sigma$.
The corresponding boundary condition for $\mathcal{D}(\tilde{H})$ (where $\mathcal{D}$ denotes the domain) reduces to,
\[\mathrm{disc}_{\Sigma}\left(\frac{\partial \psi}{\partial y}\bigg/\psi\right) = \mathrm{disc}_{\Sigma}(b_y),\quad \psi\in\mathcal{D}(\tilde{H}).\]
The putative ground state for $\tilde{H}_{\epsilon}= \hf(-\epsilon^4\Delta+\vec{b}^2+\epsilon^2\nabla\cdot\vec{b})$ is $\tilde{\psi}_{\epsilon} = \exp(R_{\epsilon}-S_{\epsilon})$, which satisfies our boundary condition.

Assuming that there is a spectral gap, when the limiting  $\tilde{H}_{\epsilon}$ is self adjoint and satisfies this boundary condition with a genuine ground state $\tilde{\psi}_{\epsilon}$   we obtain,
\[\rho_{\epsilon}(\vec{x},t)\rightarrow c_{\epsilon}\exp(2 R_{\epsilon}(\vec{x})),\]
as $t\rightarrow\infty$,
where
\[c_{\epsilon} = \frac{\iint\rho_{\epsilon}^0(\vec{x})\exp(-2S_{\epsilon}(\vec{x}))\di x \di y}{\iint\exp(2(R_{\epsilon}-S_{\epsilon})(\vec{x}))\di x \di y}.\]
This suggests that the Kepler ellipse emerges in the infinite time limit.

%%%%%%%%%%%%%%%%%%%%%%%%%%%%%%%%%%%%%%%%%%%%%%%%%%%%%%%%%%%%%%%%%%%%%%%%%%%%%
\subsection{The limiting diffusion  as $\epsilon\rightarrow 0$.}
Using the methods of Veretenikov \cite{MR890945}, and Blanchard and Golin \cite{MR882809}, we can prove the existence and uniqueness of solutions to the It\^{o} stochastic differential equations defining the limiting Nelson diffusion $\vec{X}^{\epsilon}$ even though it has a singular drift field.
However, we cannot easily control the limit of $\vec{X}^{\epsilon}$ as $\epsilon\rightarrow 0$.

We must therefore restrict $\vec{X}^{\epsilon}$ to sample paths  avoiding the singularity $\Sigma$.
We do this by  estimating the probability that up to a fixed time $t$ the process $\vec{X}^{\epsilon}(s)$ avoids the interior of a small ellipse surrounding the singularity.
Then, using the methods of Freidlin and Wentzell \cite{MR1652127}, we show how to obtain the underlying deterministic system $\vec{X}^0$ as the limit of $\vec{X}^{\epsilon}$ as $\epsilon\rightarrow 0$.

\begin{defn} For small $\delta>0$, define the first hitting time,
\[\tau_{\svec{x}}(\mathrm{int}(\mathcal{E}_{1-\delta})) = \inf\left\{s>0:\vec{X}^{\epsilon}(s)\in\mathrm{int}(\mathcal{E}_{1-\delta}),\,{\vec{X}}^{\epsilon}(0)=\vec{x}\right\},\]
where $\mathcal{E}_{1-\delta}$ is a small ellipse which surrounds the singularity $\mathcal{E}_1$.
\end{defn}

Since $\mathrm{int}(\mathcal{E}_{1-\delta})$ is open, it follows that,
\[\pr{\tau_{\svec{x}}(\mathrm{int}(\mathcal{E}_{1-\delta}))>t} = \lim\limits_{\kappa\uparrow \infty} \ex{-\kappa\int_0^t\chi(\mathrm{int}(\mathcal{E}_{1-\delta}))(\vec{X}^{\epsilon}(s))\di s}.\]
Moreover, for $\vec{x}\in\mathrm{ext}(\mathcal{E}_{1-\delta})$ the drift $\vec{b}(\vec{x})$ is Lipschitz continuous in space.
The similarity transform in Lemma \ref{transform} leads naturally to the following conjecture for small $\epsilon^2$.

\begin{conj}
\[\pr{\tau_{\svec{x}}(\mathrm{int}(\mathcal{E}_{1-\delta}))>t}\!
\approx\!
\exp\left(-(R_{\epsilon}+S_{\epsilon})\right)\left(\exp\left(-t\tilde{H}_{\epsilon}^D/\epsilon^2\right)\exp(R_{\epsilon}+S_{\epsilon})\right)(\vec{x}),\]
where,
\[\tilde{H}_{\epsilon}^D = \lim\limits_{\kappa\uparrow\infty}\left(\tilde{H}_{\epsilon}+\kappa\chi(\mathrm{int}(\mathcal{E}_{1-\delta}))\right),\]
is the Dirichlet Hamiltonian with Dirichlet boundary conditions on ellipse $\mathcal{E}_{1-\delta}$, with,
\[\tilde{H}_{\epsilon} = \hf\left(-\epsilon^4\Delta+\epsilon^2\nabla\cdot\vec{b}+\vec{b}^2\right),\]
for $\vec{b}  = \epsilon^2\nabla(R_{\epsilon}+S_{\epsilon})$, $\tilde{H}_{\epsilon}$ being the self adjoint extension of $\tilde{H}_{\epsilon}$ with domain $C_0^{\infty}(\mathbb{R}^2\setminus\Sigma)$.
\end{conj}

For our desired result we want there to be a spectral gap for $\tilde{H}_{\epsilon}^{D}$. In this connection it is worth noting that if $\nabla\cdot\vec{b}$ is bounded below and $\vec{b}^2$ has a unique global minimum at $\vec{x}_{\mathrm{min}}$, then,
 \[\tilde{H}_{\epsilon}\sim\hf\left(-\epsilon^4\Delta +\epsilon^2 \nabla\cdot\vec{b}(\vec{x}_{\mathrm{min}})+\vec{b}^2\right),\]
 and according to a celebrated result of Simon \cite{MR708966, MR750717, MR772619, MR795520}, $\sigma(\tilde{H}_{\epsilon}/\epsilon^2) \sim \sigma (H_0)$ where,
\[\quad H_0 = \hf\Delta + \hf \frac{\partial^2\tilde{V}}{\partial x_i \partial x_j}(\vec{x}-\vec{x}_{\mathrm{min}})_i (\vec{x}-\vec{x}_{\mathrm{min}}) _j+\hf\nabla\cdot\vec{b}(\vec{x}_{\mathrm{min}})+\tilde{V}(\vec{x}_{\mathrm{min}}), \]
with $\tilde{V} = \vec{b}^2/2$.

For the limiting diffusion $\vec{X}^{\epsilon}$,
\begin{equation}\label{divb}\nabla_{\svec{x}}\cdot\vec{b} = \frac{\mu(e+u) \left(e u+(e+u) (\cos v+\sin v)+\sqrt{(1-e^2) (1-u^2)}+1\right)}{4 a e\lambda (u \cos v-1) (e u+(e+u) \cos v+1)},
\end{equation}
(where $-e<u<1$) and,
\begin{equation}\label{speed}
\vec{b}^2 =\frac{\mu
   (e \cos v+1)\left(\sqrt{\left(1-e^2\right) \left(1-u^2\right)}-1\right) }{ae^2 (u \cos v-1)}.
   \end{equation}
   Moreover, from equation (\ref{speed}), $\vec{b}^2$ is symmetric about the $y$ axis ($v\mapsto -v$) and so is continuous across the singularity $\Sigma$ excluding the point $|\vec{x}|=0$. It can be  shown that in polar coordinates, $\vec{b}^2 = \mathrm{O}(r^{-\hf})$ as $r\sim 0$ uniformly for $\theta\in(0,2\pi)$ and also that $\vec{b}^2$ has a unique global minimum when $u=e^2/(2-e^2)$ and $v=\pi$ corresponding to the point,
   \[\vec{x}_{\mathrm{min}}= \left(\frac{-4a}{(1+e)(2-e)},0\right),\] at which,
    \[\vec{b}^2(\vec{x}_{\mathrm{min}})= (1-e)\frac{\mu}{2a}.\]

   However, from equation (\ref{divb}), it can be shown that  $\nabla\cdot\vec{b}$ is not so well behaved. It has a jump discontinuity across $\Sigma$ (where u=1) and also blows up at each end of $\Sigma$ (where $v=0$ and  $v=\pi$). Firstly, in polar coordinates,  $r \nabla\cdot\vec{b}\rightarrow -1$ as $r\rightarrow 0$ uniformly in $\theta$, but also $|\nabla\cdot\vec{b}|\sim r_{q}^{-1/2}$ as $r_q \rightarrow 0$ but not uniformly in $\theta$ where $r_q= |\vec{x}-( -4 a e /(1+e),0)|$ ,  meaning that the divergence is not bounded below at the points $(x,y) =( -4 a e /(1+e),0)$ and $(x,y) = (0,0)$.
Nevertheless, we hope to publish a proof of the above conjecture in the near future using the above and results of  Wang \cite{Wang}.

Now take a sequence of real numbers $\epsilon_j$ such that $\epsilon_j\rightarrow 0$ as $j\rightarrow\infty$. Then define the sequence of limiting Nelson diffusions $ \vec{X}^{\epsilon_j}$ by,
\[
\di \vec{X}^{\epsilon_j}(s)= \vec{b}( \vec{X}^{\epsilon_j}(s))\di s+\epsilon_j\di \vec{B}(s),\quad s\in(0,t),\quad  \vec{X}^{\epsilon_j}(0) =\vec{x},
\]
and the process $\vec{X}^0$ by,
\begin{equation}\label{zero_process}
\di\vec{X}^0(s) = \vec{b}(\vec{X}^0(s))\di s,\quad s\in(0,t),\quad \vec{X}^0(0) =\vec{x},
\end{equation}
where $\vec{b}$ is as in equations (\ref{bb2d}).
Set,
\[A^j_t = \left\{\omega:\quad\tau_{\svec{x}}^j(\mathrm{int}(\mathcal{E}_{1-\delta}))>t\right\},\]
where $\tau_{\svec{x}}^j$ is the first hitting time for $\vec{X}^{\epsilon_j}$. Then we have the following result which is essentially due to Freidlin and Wentzell \cite{MR1652127}:
\begin{theorem}
If $\sum\epsilon_j^2<\infty$, then,
\[\pr{\left.\vec{X}^{\epsilon_j}(s)\rightarrow\vec{X}^0(s),\, j\rightarrow\infty, \mbox{ uniformly }s\in(0,t)\right|A^j_t}=1.\]
\end{theorem}
\begin{proof} By definition,
\[\vec{X}^{\epsilon_j}(u)-\vec{X}^0(u) = \int_0^u \left(\vec{b}(\vec{X}^{\epsilon_j}(s))-\vec{b}(\vec{X}^0(s))\right)\di s +\epsilon_j \vec{B}(u),\quad u\in(0,t).\]
If we restrict $\omega\in A_t^j$, then for some  Lipschitz constant $K>0$,
\[|\vec{X}^{\epsilon_j}(u)-\vec{X}^0(u)|\le K\int_0^u |\vec{X}^{\epsilon_j}(s)-\vec{X}^0(s)|\di s+\epsilon_j|\vec{B}(u)|.\]
Set,
\[f(t) = \int_0^t|\vec{X}^{\epsilon_j}(u)-\vec{X}^0(u)|\di u,\]
then,
\[
\dot{f}(u)   = | \vec{X}^{\epsilon_j}(u)-\vec{X}^0(u)|\le K f(u) + \epsilon_j |\vec{B}(u)|,\quad u\in(0,t).
\]
Therefore,
\[\frac{\di}{\di u}\left(\exp({-Ku})f(u)\right)\le \epsilon_j\exp({-Ku})|\vec{B}(u)|,\]
and so,
\begin{eqnarray*}
f(s) &\le &\epsilon_j \exp({Ks}) \int_0^s \exp({-Ku})|\vec{B}(u)|\di u,\\
\dot{f}(s) & \le & \epsilon_j K \exp({Ks}) \int_0^s \exp({-Ku})|\vec{B}(u)|\di u+\epsilon_j|\vec{B}(s)|.
\end{eqnarray*}
Therefore,
\[\sup\limits_{s\in(0,t)} | \vec{X}^{\epsilon_j}(u)-\vec{X}^0(u)| \le 3\epsilon_j\sup\limits_{s\in(0,t)}|\vec{B}(u)|.\]
For any constant $c>0$,
\begin{eqnarray*}
\pr{\sup\limits_{s\in(0,t)} | \vec{X}^{\epsilon_j}(u)-\vec{X}^0(u)| >c}& \le & \pr{\sup\limits_{s\in(0,t)}|\vec{B}(u)|>\frac{c}{3\epsilon_j}}\\
&\le & 2 \exp\left(-\frac{c^2}{18\epsilon_j^2 t}\right).
\end{eqnarray*}
Therefore, since $\exp({-x})<x^{-1}$ for $x>0$, for any $c>0$,
\[\sum_j \pr{\sup\limits_{s\in(0,t)} | \vec{X}^{\epsilon_j}(u)-\vec{X}^0(u)| >c} \le \frac{36t}{c^2} \sum_j\epsilon_j^2<\infty.\]
The result now follows from the Borel-Cantelli lemma.
\end{proof}

The last result is vacuous unless $\pr{A_t^j}>0.$ We can estimate this probability with Jensen's inequality giving,
\[\ex{\exp\left(-\kappa\int_0^t \chi(\mathrm{int}(\mathcal{E}_{1-\delta}))(\vec{X}^{\epsilon_j}_s)\di s\right)}\qquad\qquad\qquad\qquad\qquad\qquad\qquad \]
\[\qquad\qquad\qquad\qquad\qquad\qquad\ge \exp\left(-\kappa\ex{\int_0^t\chi(\mathrm{int}(\mathcal{E}_{1-\delta}))(\vec{X}^{\epsilon_j}_s)\di s}\right).\]
If we choose $\delta$ such that the Lebesgue measure, Leb$(\mathrm{int }(\mathcal{E}_{1-\delta}))=h_t(\epsilon_j)\kappa^{-1}$, and let $\kappa\uparrow\infty$ then formally,
\[\mathbb{P}_{\svec{x}}\left\{A_t^j\right\}\ge \exp\left(-h_t(\epsilon_j)\int_0^t\di s \sup\limits_{\genfrac{}{}{0pt}{}{s\in(0,t)}{\svec{y}\in\Sigma}} p^j_s(\vec{x},\vec{y})\right),\]
where $p_s^j$ is the transition density for $\vec{X}^{\epsilon_j}$ and the supremum is taken over $\vec{y}\in\Sigma^{\pm}$, the upper and lower parts of the singularity, with $p^j_s$ possibly discontinuous across $\Sigma$.  

By methods  of Wang \cite{Wang}, we have,
\[\sup\limits_{x,y} p^j_t(\vec{x},\vec{y}) \le \frac{1}{t}\exp\left( c\left(\frac{1}{\epsilon_j^2}+\frac{t}{\epsilon_j^4}\right) \right),\]
where $\epsilon_j\in(0,1]$ and $t>0$.
Then choosing $h_t(\epsilon) = \epsilon\exp\left(-c/\epsilon^2\right) \exp\left(-ct/\epsilon^4\right)$ gives,
\[\pr{\tau_{\svec{x}}^j>t}+\pr{\tau_{\svec{x}}^j<\Delta}\le \left(\frac{\Delta}{t}\right)^{\epsilon}.\]
Moreover, the term $\pr{\tau_{\svec{x}}^j<\Delta}$ can be made arbitrarily small compared to $\pr{\tau_{\svec{x}}^j>t}$ for small $\Delta$, giving the desired result.

%%%%%%%%%%%%%%%%%%%%%%%%%%%%%%%%%%%%%%%%%%%%%%%%%%%%%%%%%%%%%%
%%%%%%%%%%%%%%%%%%%%%%%%%%%%%%%%%%%%%%%%%%%%%%%%%%%%%%%%%%%%%%
\section{The Keplerian dynamical system}
\subsection{The dynamical system and Kepler's laws of motion}
In the last section we showed that for paths avoiding the singularity,  the limiting Nelson diffusion $\vec{X}^{\epsilon}$ converged to the underlying deterministic system $\vec{X}^0$ as $\epsilon\rightarrow 0$. We now consider the behaviour of this deterministic system.
\begin{defn}\label{kepler_dynam_defn}
The Keplerian dynamical system is  given by the equations,
\[\dot{\vec{x}}=\vec{b}(\vec{x})\quad\Leftrightarrow\quad\dot{x} = b_x(x,y),\quad\dot{y}=b_y(x,y), \]
where $b_x$ and $b_y$ are as defined in equations (\ref{bb2d}).
\end{defn}

The vector field for this dynamical system is shown in Figure \ref{vecfield} together with the Kepler ellipse and the singularity $\Sigma$.  

\begin{remark} Although we use the name `dynamical system' for the differential equations in Definition \ref{kepler_dynam_defn}, they do not necessarily give a system satisfying the definition of a true dynamical system (see \cite{MR0219843, MR0289890} for a full definition). This is because the existence of the  singularity $\Sigma$ means that we cannot use the standard methods of differential equations to ascertain the existence, uniqueness and extendability of solutions through every point of $\mathbb{R}^2$.
 However, as we will show, any path starting outside the Kepler ellipse will stay away from the singularity and so we can assume that the definition applied to the exterior of the Kepler ellipse will produce a dynamical system. Where necessary we will assume that we are working solely on the exterior of $\mathcal{E}_e$. We hope to extend this work to include all paths using the work of Fillipov \cite{MR1028776}.
 \end{remark}

\begin{figure}
\centering
\resizebox{120mm}{!}{\includegraphics{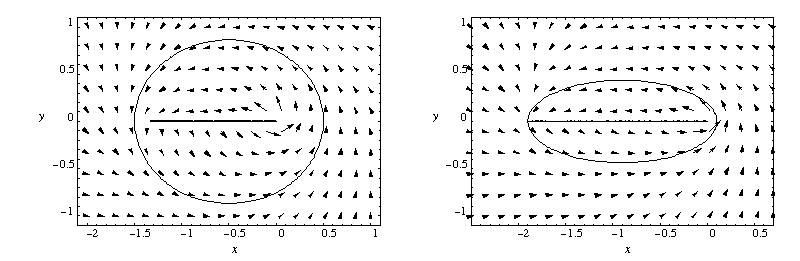}}
\caption{The vector field of the Keplerian dynamical system with $e=0.5$ and $0.9$ showing the Kepler ellipse and singularity.}
\label{vecfield}
\end{figure}

We begin our anlaysis of the Keplerian dynamical system with a very simple derivation of Kepler's laws of planetary motion.

\begin{theorem}\label{kepler_motion}
The Keplerian dynamical system has the Kepler ellipse $\mathcal{E}_e$ as a periodic orbit. Moreover, a particle moving on the periodic orbit $\mathcal{E}_e$ obeys Kepler's laws of motion with force constant $\mu$ and energy $E =- \mu/2 a$.
\end{theorem}
\begin{proof}
Using Proposition \ref{itocorrection}, we can write the deterministic system in the Keplerian elliptic coordinates as,
\begin{equation}\label{Keplerianuv}
\dot{u} = b_u(u,v),\quad \dot{v} = b_v(u,v),
\end{equation}
with $\epsilon=0$. The Kepler ellipse $\mathcal{E}_e$ has the equation $u = e$ and a simple calculation gives  $b_u(e,v) \equiv 0$ so that $\mathcal{E}_e$ is a periodic orbit.

Moreover, for $u=e$, equations (\ref{coords}) which define $u$ and $v$  reduce to give,
\[x = a \cos v -ae,\qquad y = a\sqrt{1-e^2} \sin v, \]
so that $v$ is the eccentric angle for the ellipse $\mathcal{E}_e$. Then, the equations of motion   (\ref{Keplerianuv}) reduce to,
\[\dot{v} = b_v(e,v) =\frac{\mu}{\lambda a(1 - e\cos v)} = \sqrt{\frac{\mu}{a^3}}\cdot\frac{1}{(1 - e\cos v)} ,\]
which is the Kepler law of motion in terms of the eccentric angle $v$ with force constant $\mu$ and energy $E = -\mu/2 a$, where $a = \lambda^2/\mu$.
 \end{proof}
In the following sections we will consider the convergence of trajectories of the Keplerian dynamical system to the Kepler ellipse.  In the next section we will show which initial positions produce trajectories which avoid the singularity. In the subsequent sections we will consider the asymptotic stability of the Kepler ellipse and we  will show that all trajectories which avoid the singularity will converge to the Kepler ellipse  in such a manner that they will also obey Kepler's laws in the infinite time limit.

 %%%%%%%%%%%%%%%%%%%%%%%%%%%%%%%%%%%%%%%%%%%%%%%%%%%%%%%%%%%%%%%%%%%%%%%
\subsection{Avoiding the singularity}
We want to show which initial positions give rise to trajectories that 
avoid the singularity.
For the dynamical system,
\[\dot{\vec{x}}(t) = \vec{b}(\vec{x}(t)),\quad t>t_0,\quad \vec{x}(t_0) = \vec{x}_0,\]
we denote the solution $\vec{x}(t) = \vec{x}(t;\vec{x}_0,t_0)$.

From equation (\ref{Keplerianuv}), the Keplerian dynamical system can be written in Keplerian elliptic coordinates as,
\[\dot{u} = b_u(u,v),\quad \dot{v} = b_v(u,v).\]
The drift $b_u$ tells us about the motion towards and away from the singularity $\Sigma $ which is given by $u=1$.
Thus, we are particularly interested in the sign of $b_u$.

\begin{lemma}\label{curlyf}
For the Keplerian dynamical system represented in Keplerian elliptic coordinates, if $u\in(-e,1)$,
\[
b_u = 0 \quad \Leftrightarrow\quad
\left\{
\begin{array}{ll}
u=e,  &v\in(0,2 \pi),\mbox{ or,}\\
u=\mathcal{F}(v),   & v\in (\pi/2,\pi),
\end{array}\right.
 \]
 where
 $\mathcal{F}(v) =\mathcal{F}_{(e,1)}(v)/\mathcal{F}_{(-1,-e)}(v)$
 and,
 \begin{equation}\label{ffunction}
 \mathcal{F}_{(e_1,e_2)}(v) =  e_1(1-\cos v\sin v)+e_2(\cos v-\sin v).
 \end{equation}
\end{lemma}
\begin{proof}
A simple calculation from Proposition \ref{itocorrection}. The curve $u=\mathcal{F}(v)$ is shown in the Cartesian frame in Figure \ref{bu_sign_pic}.
\end{proof}

 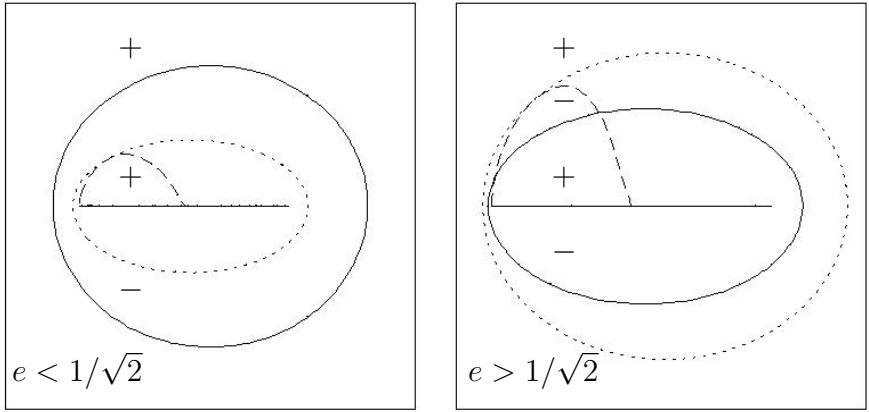
\begin{figure}
\centering
\setlength{\unitlength}{1cm}
\begin{picture}(12,5)
\put(0,0){\resizebox{120mm}{!}{\includegraphics{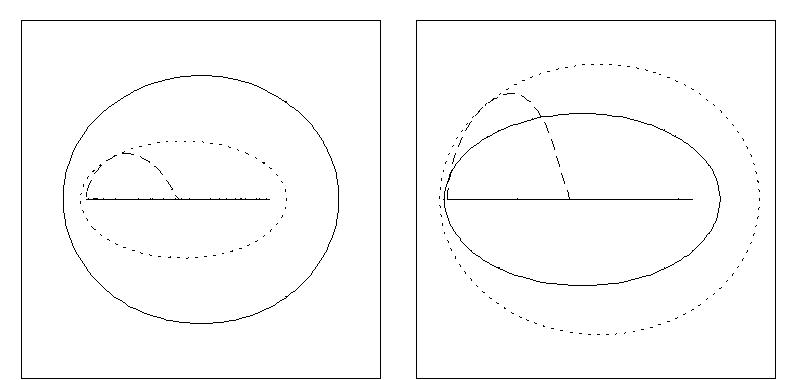}}}
\put(0.4,0.4){$e<1/\sqrt{2}$}
\put(1.8,4.7){$+$}
\put(1.8,1.5){$-$}
\put(1.8,3){$+$}
\put(6.4,0.4){$e>1/\sqrt{2}$}
\put(7.5,4.7){$+$}
\put(7.5,2){$-$}
\put(7.5,3){$+$}
\put(7.5,4){$-$}
\end{picture}
\caption{The Kepler ellipse $\mathcal{E}_e$ and singularity (solid line) with the curve $u=\mathcal{F}(v)$ (dashed line), the ellipse $\mathcal{E}_{\tilde{e}}$ (dotted line) and the sign of $b_u$ indicated.}
\label{bu_sign_pic}
\end{figure}

\begin{cor}\label{crash}
The singularity $\Sigma$ in the cartesian plane is repulsive everywhere except for,
\[-\frac{4 a e}{1+e} \le x\le -\frac{2 a e}{1+e},\]
where for $y>0$ the singularity is attractive.
\end{cor}
\begin{proof}
For $u=1$, $b_u>0$ for $v\in(\pi/2,\pi)$ but $b_u<0$ for all other $v$.
\end{proof}

Corollary \ref{crash} shows that paths which start away from the singularity may reach it in a finite time. Using Lemma \ref{curlyf} we can identify which initial positions may lead to this behaviour. It can be shown that the curve $u=\mathcal{F}(v)$ reaches its minimum at $v = 3\pi/4$ and its maximum at $v=\pi/2$ and $v=\pi$. Therefore,
\[1>\mathcal{F}(v)>\tilde{e},\quad\mbox{ where}\quad \tilde{e} := \frac{-2\sqrt{2} +3 e}{-3+2\sqrt{2} e}.\] 

\begin{cor}
If $\vec{x}_0\in \mathrm{ext}(\mathcal{E}_e)\cup  \mathrm{ext}(\mathcal{E}_{\tilde{e}})$ then $\vec{x}(t;\vec{x}_0,t_0)$ will never reach the singularity for $t>t_0$.
\end{cor}
This tells us that outside an ellipse surrounding the singularity $\Sigma = \mathcal{E}_1$, the trajectories of the dynamical system will not intersect $\Sigma$.
The ellipses $\mathcal{E}_e$ and $\mathcal{E}_{\tilde{e}}$ are shown in Figure \ref{bu_sign_pic}.

\begin{remark}  From this we can conclude the following (see Figure \ref{bu_sign_pic}):
\begin{enumerate}
\item
For $e=1/\sqrt{2}$ it follows that $\tilde{e} = 1/\sqrt{2}$ so that $\mathcal{E}_e = \mathcal{E}_{\tilde{e}}$.
\item
 For $e<1/\sqrt{2}$ it follows that $1>\tilde{e} >e $ so that $ \mathcal{E}_{e}\subset \mathrm{ext}( \mathcal{E}_{\tilde{e}}) $.
  \item
 For $e>1/\sqrt{2}$ it follows that $-e<\tilde{e} <e$ so that $ \mathcal{E}_{\tilde{e}}\subset \mathrm{ext}(\mathcal{E}_e) $.
\end{enumerate}
\end{remark}

 %%%%%%%%%%%%%%%%%%%%%%%%%%%%%%%%%%%%%%%%%%%%%%%%%%%%%%%%%%%%%%%%%%%%%%%
\subsection{The asymptotic stability of Keplerian motion}
 We will now look at the attraction of paths which avoid the singularity to the Keplerian orbit and also at the stability of this orbit.

\begin{defn}
A periodic orbit generating a closed trajectory $C$ is called orbitally stable if for any $\delta_1>0$ and any initial position $\vec{x}_0$ which yields a periodic solution traversing $C$, there exists $\delta_2>0$ such that,
\[|\tilde{\vec{x}}_0-\vec{x}_0|<\delta_2\quad\Rightarrow\quad
d(\vec{x}(t;\tilde{\vec{x}}_0,t_0),C)<\delta_1,\]
for $t>t_0$, where $d$ denotes the metric
$d(\vec{x},C) = \inf|\vec{x}-\vec{y}|$ for $\vec{y}\in C$.

Moreover, if there also exists a $\delta_3>0$ such that,
\[|\tilde{\vec{x}}_0-\vec{x}_0|<\delta_3 \quad\Rightarrow\quad\lim\limits_{t\rightarrow\infty}d(\vec{x}(t;\tilde{\vec{x}}_0,t_0),C)=0,\]
then the orbit is said to be asymptotically orbitally stable.
\end{defn}

For a two dimensional system we have the following result \cite{Willems}.

\begin{theorem}
The closed trajectory $C$ corresponding to a periodic solution $\vec{x}_p(t)$ of period $T$ of a dynamical system is asymptotically orbitally  stable if,
\[\int_0^T\mathrm{Tr}\left(\left.\frac{\partial \vec{b}}{\partial\vec{x}}(\vec{x})\right|_{\vec{x}=\vec{x}_p(t)}\right)\di t<0.\]
\end{theorem}
\begin{theorem}
For the Keplerian dynamical system let $\vec{x}_p(t)$ be a periodic trajectory  traversing the Kepler ellipse $\mathcal{E}_e$ with period $T=2\pi\sqrt{a^3/\mu}$. Then,
\[\int_0^T\mathrm{Tr}\left(\left.\frac{\partial \vec{b}}{\partial\vec{x}}(\vec{x})\right|_{\vec{x}=\vec{x}_p(t)}\right)\di t=-2\pi.\]
\end{theorem}
\begin{proof}
We can calculate this integral using Keplerian elliptic coordinates,
\begin{eqnarray*}
\int_0^T\mathrm{Tr}\left(\left.\frac{\partial \vec{b}}{\partial\vec{x}}(\vec{x})\right|_{\vec{x}=\vec{x}_p(t)}\right)\di t
& = & \int_0^{2\pi}\left\{ \frac{1}{b_v} \left(\frac{\partial b_x}{\partial x}+\frac{\partial b_y}{\partial y}\right)\right\}_{u=e}\di v\\
& = & -\int_0^{2\pi} \frac{e \cos v+e \sin v+1}{e^2+2 e\cos v +1}\di v\\
& = & -2 \int_0^{\pi} \frac{e \cos v+1}{e^2+2 e\cos v +1}\di v\\
& = &- \lim\limits_{v\uparrow \pi} \left(v+2 \tan ^{-1}\left(\frac{1-e}{1+e} \tan
   \left(\frac{v}{2}\right)\right)\right)\\
   & = & - 2 \pi
\end{eqnarray*}
as $1>e>0$.
\end{proof}
\begin{cor}\label{orbital_stab}
The Kepler ellipse $\mathcal{E}_e$ is an asymptotically orbitally stable periodic orbit for the Keplerian dynamical system.
\end{cor}
We can also demonstrate that the Kepler ellipse is asymptotically orbitally stable by using the invariant density $\rho_{\epsilon}^{\infty}$ from equation (\ref{invariantdensity}) as  a Lyapunov function.
\begin{theorem}\label{lyapunov}
Let $V_{\mathrm{Lpv}}(\vec{x})$ be a real valued function defined in an open neighbourhood $N(C)$ of a compact set $C$. Assume that,
\begin{enumerate}
\item $V_{\mathrm{Lpv}}(\vec{x})$ is continuously differentiable,
\item $V_{\mathrm{Lpv}}(\vec{x})$ is positive definite in $N(C)\setminus C$,
\item $\nabla V_{\mathrm{Lpv}}(\vec{x})\cdot \vec{b}(\vec{x})$ is negative definite in $N(C)\setminus C$,
\item $V_{\mathrm{Lpv}}(\vec{x})=\nabla V_{\mathrm{Lpv}}(\vec{x})\cdot \vec{b}(\vec{x})=0$ for $\vec{x}\in C$.
\end{enumerate}
Then the compact set $C$ is asymptotically stable.
\end{theorem}
The function $V_{\mathrm{Lpv}}$ in Theorem \ref{lyapunov} is called a Lyapunov function.
\begin{theorem}\label{lyapunovthm}
The function,
\[V_{\mathrm{Lpv}} =\left(16/e^{2}\right)^{\lambda}-\left.\exp(2R_{\epsilon}) \right|_{\epsilon=1},\]
satisfies the conditions of Theorem \ref{lyapunov} for the compact set $C = \mathcal{E}_e$ and the neighbourhood $N(C) =\mathrm{int}( \mathcal{E}_{-e+\delta_1})\setminus\mathrm{int}(\mathcal{E}_{1-\delta_2})$ for any $\delta_1,\delta_2\in(0,1+e)$ such that $\delta_1+\delta_2<1+e$.
\end{theorem}
\begin{proof}
This follows simply from Theorem \ref{invariantthm}. The function $V_{\mathrm{Lpv}}$ is shown in Figure \ref{liap1}. It is continuous but is not smooth across $\Sigma$. The derivative $\nabla V_{\mathrm{Lpv}}\cdot\vec{b}$ is shown in Figure \ref{liap2}. We exclude the ellipse at infinity $\mathcal{E}_{-e}$, as the derivative is zero in the limit.
\end{proof}

 \begin{figure}
\centering
\resizebox{120mm}{!}{\includegraphics{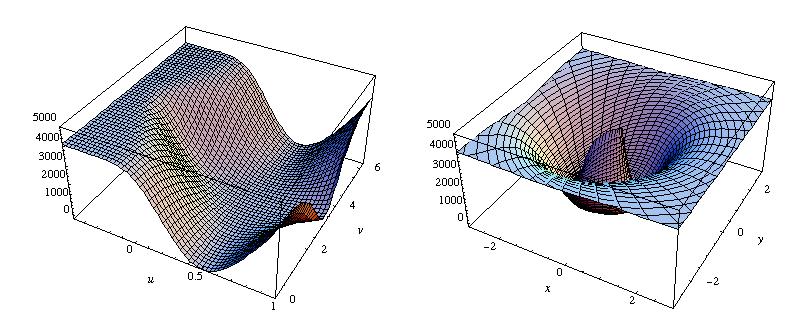}}
\caption{The Lyapunov function $V_{\mathrm{Lpv}}$ in $(u,v)$ and $(x,y)$ coordinates for $e=0.5$.}
\label{liap1}
\end{figure}
 \begin{figure}
\centering
\resizebox{120mm}{!}{\includegraphics{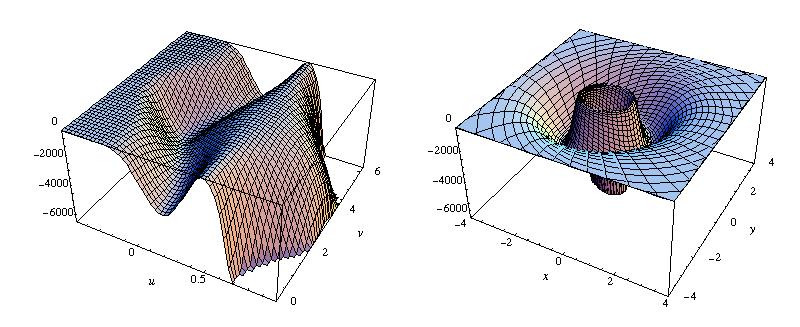}}
\caption{The derivative $\nabla_{\svec{x}}V_{\mathrm{Lpv}}\cdot\vec{b}(\vec{x})$ in $(u,v)$ and $(x,y)$ coordinates for $e=0.5$.}
\label{liap2}
\end{figure}

Asymptotic orbital stability tells us that not only does any trajectory sufficiently close to $\mathcal{E}_e$ converge in the large time limit to $\mathcal{E}_e$ but that in this large time limit any converging trajectory will move exactly as a trajectory on $\mathcal{E}_e$. Therefore, the motion of a particle which converges towards the ellipse $\mathcal{E}_e$ must converge to Keplerian motion on $\mathcal{E}_e$.

 %%%%%%%%%%%%%%%%%%%%%%%%%%%%%%%%%%%%%%%%%%%%%%%%%%%%%%%%%%%%%%%%%%%%%%%
\subsection{Region of attraction to the Kepler ellipse}

We now look at which initial positions produce trajectories which are attracted to the Kepler ellipse in the infinite time limit. As in Section 6.2 we are particularly interested in the sign of $b_u$ as the Kepler ellipse is given by the coordinate curve $u=e$ and we know that $b_u=0$ on this curve.
\begin{theorem}\label{smalle}
Let $0<e<1/\sqrt{2}$. For any initial position $\vec{x}_0\in \mathrm{ext}(\mathcal{E}_{\tilde{e}})$ and any given $\delta>0$ there exists a finite time $T>0$ such that,
\[d(\vec{x}(t;\vec{x}_0,t_0),\mathcal{E}_e)<\delta,\]
for all times $t>T$.
\end{theorem}
\begin{proof}
For the curve $u=\mathcal{F}(v)$  to cut the Kepler ellipse $u=e$ we would require $e<\tilde{e}$ or equivalently $e<1/\sqrt{2}$.
Therefore, if $0<e<1/\sqrt{2}$ then (as shown in Figure \ref{bu_sign_pic}):
\begin{enumerate}
\item for $-e<u<e$, $b_u >0$,
\item for $ e<u<\mathcal{F}(v)$, $b_u<0$.
\end{enumerate}
Since $\vec{b}^2>0$ for all $u\in(-e,1)$ and $v\in[0,2\pi)$ the result follows.
\end{proof}

From the proof of Theorem \ref{smalle} we see that for $0<e<1/\sqrt{2}$, the Kepler ellipse is a stable orbit in the $(u,v)$ frame as the curve $u=\mathcal{F}(v)$ does not intersect the Kepler ellipse. However, when $e=1/\sqrt{2}$ the curve $u=\mathcal{F}(v)$ touches the Kepler ellipse and for $e>1/\sqrt{2}$ it intersects the Kepler ellipse. Thus, there is a portion of the Kepler ellipse which is unstable in this frame. Despite this we cannot conclude that the Kepler ellipse will be unstable in the cartesian frame as the properties of stability and instability are not necessarily coordinate invariant. As we showed in Corollary \ref{orbital_stab}, the Kepler ellipse  is actually stable. As can be seen in Figure \ref{det_trajs}, simulations of the dynamical system starting outside the Kepler ellipse always converge to the ellipse. The region bounded by the curve $u=\mathcal{F}(v)$ is where this convergence is particularly slow. It would seem sensible to investigate whether $e=1/\sqrt{2}$ is a critical value of the eccentricity. It could be that this apparent instability is a property of the $(u,v)$ coordinate system. As we shall see in the next section this is not the case!

If we assume $e<1/\sqrt{2}$ then we can avoid these problems and combining  our results we can conclude:

\begin{theorem}\label{lessthanroot2}
For $0<e<\frac{1}{\sqrt{2}}$, any orbit of our Keplerian dynamical system with a start point outside the ellipse $\mathcal{E}_{\tilde{e}}$ settles down in the infinite time limit to Keplerian motion on the Kepler ellipse $\mathcal{E}_e$.
\end{theorem}
\begin{proof}
This follows from asymptotic stability combined with Theorems \ref{kepler_motion} and \ref{smalle}.
The region of convergence is shown in Figure \ref{two_thms} (a).
\end{proof}

 We now consider ways to extend this domain of attraction to include eccentricities $e>1/\sqrt{2}$.
 Arguing exactly as above we have the following result which restricts the location of the trajectory after a finite time to a large neighbourhood of the Kepler ellipse.

\begin{theorem}\label{morethanroot2}
Let $1>e>1/\sqrt{2}$. For any initial position $\vec{x}_0\in \mathrm{ext}(\mathcal{E}_{e})$ there exists a finite time $T>0$ such that,
\[\vec{x}(t;\vec{x}_0,t_0)\in \mathrm{int}(\mathcal{E}_{\tilde{e}})\cap\mathrm{ext}(\mathcal{E}_e),\]
for all times $t>T$.
\end{theorem}
\begin{proof}
Follows as for Theorem \ref{smalle} but now $\tilde{e} >e$ so that $\mathcal{E}_e\subset\mathcal{E}_{\tilde{e}}$. This region is shown in Figure \ref{two_thms}(b).
\end{proof}
 \begin{figure}
\centering
\resizebox{120mm}{!}{\includegraphics{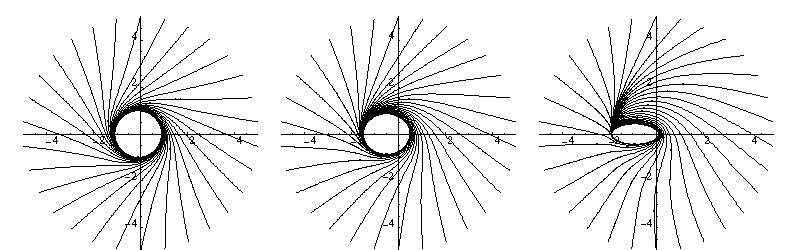}}
\caption{Simulations of the Keplerian dynamical system with $e=0.1$, $0.5$ and $0.9$.}
\label{det_trajs}
\end{figure}

 \begin{figure}
\centering
\setlength{\unitlength}{1cm}
\begin{picture}(12,5)
\put(0,0){\resizebox{120mm}{!}{\includegraphics{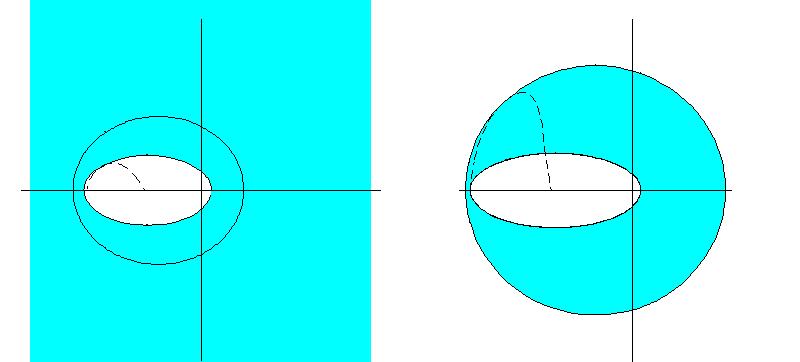}}}
\put(0.4,0.4){(a)}
\put(1.8,4){$\mathcal{E}_e$}
\put(2.3,2.35){$\mathcal{E}_{\tilde{e}}$}
\put(6.4,0.4){(b)}
\put(9,4.8){$\mathcal{E}_{\tilde{e}}$}
\put(8.7,2.35){$\mathcal{E}_{{e}}$}
\end{picture}
\caption{Theorems \ref{lessthanroot2} and \ref{morethanroot2}: (a) The region of attraction (shaded) for $\mathcal{E}_e$ with $e<1/\sqrt{2}$.
(b) For $e>1/\sqrt{2}$ any trajectory starting outside the Kepler ellipse must converge to somewhere in the shaded region.}
\label{two_thms}
\end{figure}
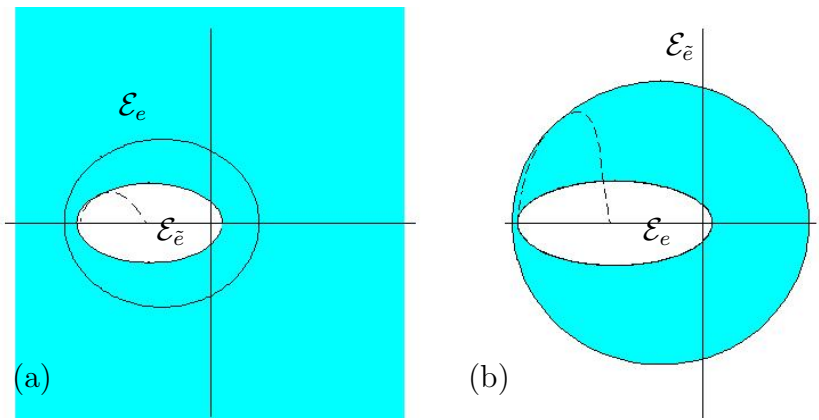

This still leaves us to prove that trajectories in this elliptic annulus converge to the Kepler ellipse. This can be done using the Lyapunov function in Theorem \ref{lyapunovthm} and results of \cite{MR0219843}.

%%%%%%%%%%%%%%%%%%%%%%%%%%%%%%%%%%%%%%%%%%%%%%%%%%%%%%%%%%%%%%%%%%%%%%%
 \subsection{Surprising symmetries and mild instabilities}
 The appearance of $e=1/\sqrt{2}$ as an important eccentricity  could easily be dismissed as a property of the Keplerian elliptic coordinate system. Numerical experiments confirm that although $b_u$ does switch sign in a neighbourhood of the Kepler ellipse for $e>1/\sqrt{2}$, the  Cartesian dynamical system still always moves towards the Kepler ellipse in this region. This discrepancy can be explained easily by considering the bunching of level surfaces of $u$ which occurs in  the Keplerian elliptic coordinates (see Figure \ref{coordpic}) and also the slow rate of convergence in this region (see Figure \ref{det_trajs}), particularly for large eccentricities.
It is therefore surprising to encounter this critical eccentricity through other calculations.

Recall that if $u\neq e$ then,
\[b_u=0\quad\Leftrightarrow\quad u= \frac{\mathcal{F}_{(e,1)}(v)}{\mathcal{F}_{(-1,-e)}(v)},\quad v\in(\pi/2,\pi),\]
where $\mathcal{F}_{(e_1,e_2)}(v)$ is as defined in equation (\ref{ffunction}).

If we consider the divergence of the drift field calculated in Cartesians but presented in Keplerian elliptic coordinates for simplicity (see equation (\ref{divb})) then we find,
\[\nabla\cdot\vec{b}=0\quad\Leftrightarrow\quad u = \frac{\mathcal{F}_{(e,1)}(-v)}{\mathcal{F}_{(-1,-e)}(-v)},\quad v\in(\pi,3\pi/2).\]
Therefore,  the curve $\nabla\cdot\vec{b}=0$ is the reflection in the $y$ axis of the curve $b_u =0$. Moreover, this leads to the conclusion that for $e<1/\sqrt{2}$, the curve $\nabla\cdot\vec{b}$ does not cut the Kepler ellipse, and so on all parts of the ellipse $\nabla\cdot\vec{b}<0$. However, for $e>1/\sqrt{2}$ there is a portion of the ellipse where $\nabla\cdot\vec{b}>0$ suggesting a local mild  instability in this region.

These instabilities are not just properties of the two dimensional restriction. Consider the full three dimensional system for small values of $|z|$. The drift in the $z$ direction is given by,
\[b_z = -\frac{\mu}{2\lambda}\left(\alpha+\beta+1\right)\frac{z}{|\vec{x}|}.\]
Thus, if $\alpha+\beta+1>0$, then any trajectory will  be attracted to the $z=0$ plane where it will remain stable.
Let us consider the behaviour of $\alpha+\beta+1$ in the plane $z=0$. Working from equations (\ref{alphabetauv}) it is simple to show,
\[\alpha+\beta+1=0\quad\Leftrightarrow\quad u=\frac{\mathcal{F}_{(e,-1)}(-v)}{\mathcal{F}_{(-1,e)}(-v)},\quad v\in(0,\pi/2),\]
which is a curve in the same family again. It again meets the Kepler ellipse when $e=1/\sqrt{2}$. Therefore, when $e>1/\sqrt{2}$, there will be a region of the Kepler ellipse where the orbit is unstable in the $z$ direction as within this curve $\alpha+\beta+1<0$ which will mean $b_z>0$ for small $z>0$ and $b_z<0$ for small $z<0$. These three curves are shown together in Figure \ref{symmetrypic}.

 \begin{figure}
\centering
\resizebox{120mm}{!}{\includegraphics{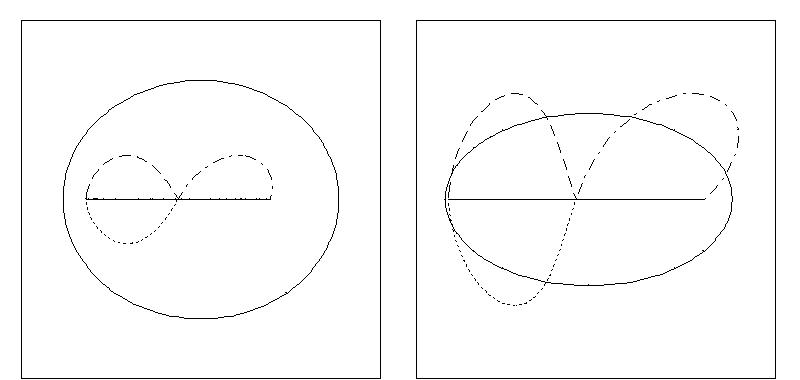}}
\caption{The curves $b_u=0$ (long dash), $\nabla\cdot\vec{b}=0$ (dotted) and $\alpha+\beta+1=0$ (dash-dot line) with the Kepler ellipse and singularity for $e=0.5$ and  $0.8$.}
\label{symmetrypic}
\end{figure}

In fact the curve $\alpha+\beta+1=0$ is a slice of a three dimensional surface  which, together with the singularity, bounds a region of space where $b_z$ will be directed away from the plane $z=0$. This creates a general instability in orbits for  $e>1/\sqrt{2}$ as this region is cut by the Kepler ellipse. The effect of this in the deterministic case is shown in Figure \ref{zbump}. There is a clear blip in the value of $z$ where the particle moves away from the plane $z=0$. This blip is repeated each time the particle passes through this region, but each time the blip becomes smaller as the particle moves closer to the plane $z=0$ overall. In the infinite time limit the particle still appears to converge to the plane $z=0$. However, in the stochastic case ($\epsilon\neq 0$) the noise creates a displacement such that the particle never appears to settle down into a small neighbourhood of the  plane $z=0$ (see Figure \ref{zbump2}). This is caused by the presence of the noise, but is apparently exacerbated by these repeated blips in $b_z$, which unlike the deterministic case do not decrease in magnitude in the infinite time limit (see Figure \ref{zbump3}). This effect becomes more pronounced as $e \rightarrow 1$ since the trajectories pass closer to the origin where $\bm{b} \sim \pm r^{-1/2}$. The periodic spikes in $b_z$ shown in Figure \ref{zbump3} are the result of this effect and cause the particle to be strongly forced into the plane $z=0$ when at the perihelion. The blip in $b_z$ occurs immediatly after the particle has passed the perihilion forcing it away from the plane $z=0$ again.

 \begin{figure}
\centering
\resizebox{130mm}{!}{\includegraphics{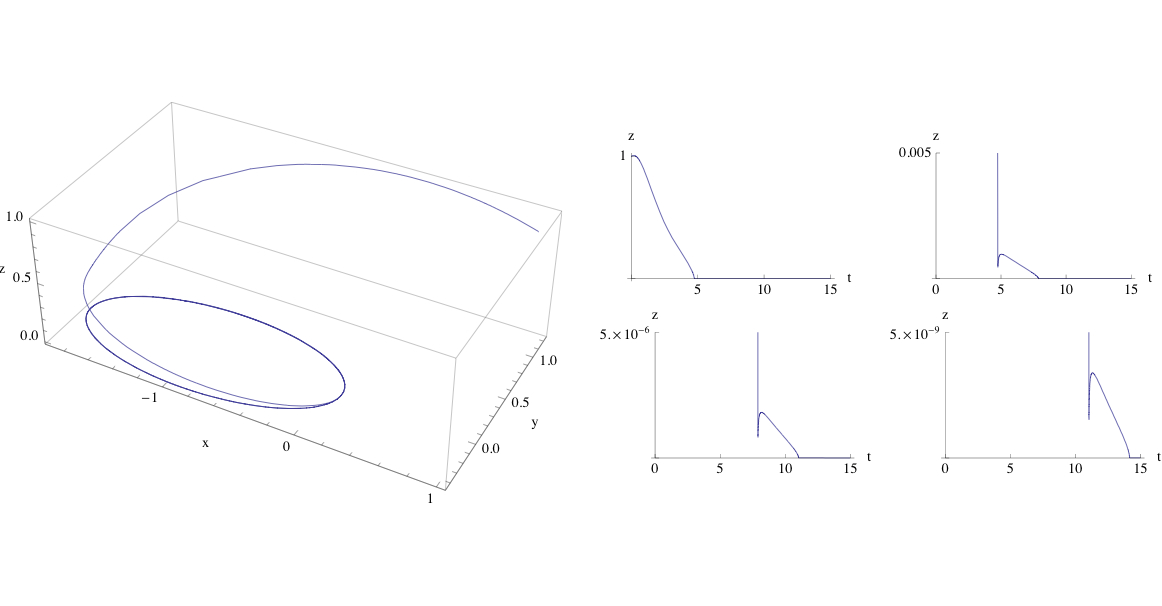}}
\caption{Simulation for $\epsilon=0$ and $e=0.9$ highlighting the periodic blips in $z$.}
\label{zbump}
\end{figure}

 \begin{figure}
\centering
\resizebox{130mm}{!}{\includegraphics{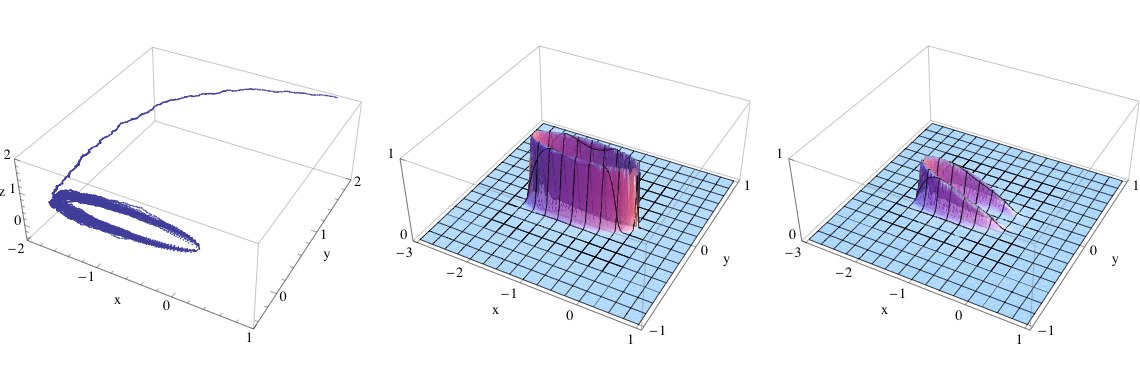}}
\caption{Simulation for $\epsilon=0.05$ and $e=0.99$ with the invariant measure.}
\label{zbump2}
\end{figure}

 \begin{figure}
\centering
\resizebox{130mm}{!}{\includegraphics{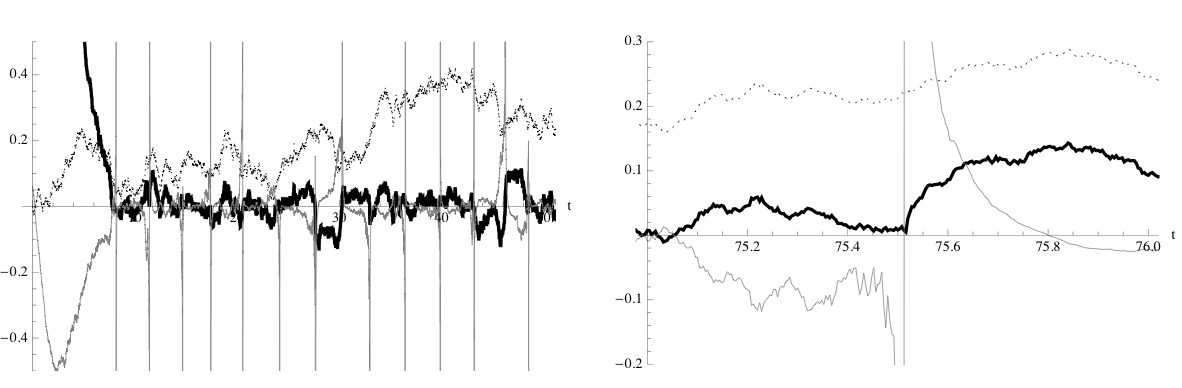}}
\caption{The $z$ value (thick) with velocity $b_z$ (thin) and driving noise $B_3$(dashed) for $\epsilon=0.05$ and $e=0.99$ with a close up of a blip.}
\label{zbump3}
\end{figure}

 %%%%%%%%%%%%%%%%%%%%%%%%%%%%%%%%%%%%%%%%%%%%%%%%%%%%%%%%%%%%%%%%%%%%%%%
  %%%%%%%%%%%%%%%%%%%%%%%%%%%%%%%%%%%%%%%%%%%%%%%%%%%%%%%%%%%%%%%%%%%%%%%
\section{Conclusions}
We have shown that in the Bohr correspondence limit in two dimensions the Nelson diffusion process corresponding to the atomic elliptic state for the Coulomb problem reduces to Keplerian motion on the Kepler ellipse for motions starting outside this ellipse. This solves a long standing problem in quantum mechanics.
In the quantum mechanical setting, Kepler's laws of planetary motion need to be augmented with a caveat about the mild local instabilities appearing on the elliptical orbit for eccentricities greater than $1/\sqrt{2}$ which may be experimentally detectable.
%These could provide an experimental test of Nelson's stochastic mechanics if the atomic elliptic state could be realised in the laboratory. They may also shed light on the way an orbiting electron interacts with the electromagnetic field in hydrogen like atoms.
In any case these results merit further study in the setting of planetesimal diffusions.

We are currently investigating the three dimensional problem which appears to be far more difficult to analyse in the absence of any convenient coordinate system.

 %%%%%%%%%%%%%%%%%%%%%%%%%%%%%%%%%%%%%%%%%%%%%%%%%%%%%%%%%%%%%%%%%%%%%%%
  %%%%%%%%%%%%%%%%%%%%%%%%%%%%%%%%%%%%%%%%%%%%%%%%%%%%%%%%%%%%%%%%%%%%%%%
\section*{Acknowledgements}
It is a pleasure for AN to thank the Welsh Institute for Mathematical and Computational Sciences (WIMCS) for their financial support in this research. 
We would also like to thank Prof. F-Y Wang of WIMCS for helpful conversations concerning parts of this work. Finally, AT would like to thank the Centro di Ricerca Matematica Ennio De Giorgi in Pisa for their hospitality during a research visit in July 2006 where some of the ideas in this work were developed.

\bibliographystyle{plain}
\def\cprime{$'$}

\end{document}